\documentclass[10pt, twosides]{amsart}

\usepackage{defs, amsaddr}

\title[Discrete time Observer for Continuous time system]{Design of a discrete time observer for the continuous time rotation kinematics on $\SO3$}
\thanks{}

\author[S. Shanbhag]{Soham Shanbhag}
\address{Department of Mechanical Engineering\\ IIT Bombay, Powai\\ Mumbai 400076, India.\\
\url{https://sohamshanbhag.github.io}\\
}

\author[R. Banavar]{Ravi Banavar}
\address{Systems \& Control Engineering\\ IIT Bombay, Powai\\ Mumbai 400076, India.\\
\url{http://www.sc.iitb.ac.in/~banavar}\\%{http://www.sc.iitb.ac.in/~banavar}\\
}
\email{soham.shanbhag@gmail.com, banavar@iitb.ac.in}


\toggletrue{arxiv}
\usepackage{import}
\usepackage{pdfpages}

\begin{document}

	\begin{abstract}
		This report proposes a discrete time observer for the continuous time rigid body kinematics on the rotation group $\SO3$. The work draws on two research schools - one by Chang \cite{Chang_FI} based on feedback integrators for systems evolving on manifolds, and the other by Mahony \cite{Mahony}, who proposed an observer for attitude dynamics.
The discrete time observer is based on the modified dynamics of the Mahony observer for attitude dynamics, where the modification of the vector field enables numerical integration based on Euclidean schemes.

        \end{abstract}

	\maketitle

	\section{Introduction}
		\label{sec:intro}
		% INTRODUCTION
%
The special orthogonal group $\SO3$ (the group of rotations) finds wide applications in mechanical and aerospace engineering systems.
A very popular and widely used application is the quadrotor, where the state of the quadrotor is partially constituted by the orientation of a body frame fixed to the quadrotor with respect to a spatial frame.
To implement control laws, it is essential to have information (or knowledge) of the state of the dynamical system.
Often due to sensor limitations or noise, only part of the state vector is accessible, and the rest have to be estimated.
Furthermore, the system state data is only available at discrete instants of time since most implementation is digital and processing occurs in discrete time.
To distinguish terminology right from the outset, in this article the word {\it estimation} is used in the context where the system and measurements are corrupted with noise, while the word {\it observer} is referred to a situation where the system is noise-free and the state-estimate is being sought from an incomplete measurement of the state.

A large body of literature is available on estimators/observers in Euclidean spaces.
\cite{Respondek} designs asymptotic observers on $\R^3$ for a class of nonlinear systems.
One of the most famous algorithms in estimation theory, proposed in \cite{Kalman1960}, provides a recursive algorithm to estimate the state of the system assuming Gaussian noise.
There have been a lot of advancements in this field based on this theory.
A continuous time version of the filter was given by \cite{kalmanbucy}.
The unscented filter was given by \cite{Julier97anew}.
However, the Kalman filter and all its extensions assume that the state belongs to a Euclidean space, which may not always be the case in many engineering applications.

The configuration variables of many mechanical and aerospace systems, like serial link
robots, satellites, quadrotors, evolve on non-Euclidean spaces or smooth manifolds.
So it is essential to develop algorithms, in particular, in discrete time, for such systems.
Often, problems on manifolds are viewed locally in terms of charts,  which can be used to map open sets on the manifold to open sets in $\R^n$, and hence develop existing
or slightly modified  filters on the mapped Euclidean space.
However, this method has multiple disadvantages.
The chart maps are only locally defined, hence the filter equations need to be checked for continuity and differentiability at every boundary of the open set and further, the results
are not global.
An extension of the Kalman filter to Riemannian manifolds is given by \cite{Hauberg2013}.
In this this work the Riemannian metric is used to derive the unscented Kalman Filter
on the manifold.
Moreover, this method requires computation of the logarithmic map which requires solving an optimal control problem.
Another Kalman Filter proposed on Lie Groups is given by \cite{bourmaud:hal-00903252}.

The efforts towards constructing observers for systems evolving on manifolds begins with papers by Bonnabel(\cite{barrau:hal-01826025}, \cite{DBLP:journals/corr/BarrauB14}, \cite{DBLP:journals/corr/BarrauB13}, \cite{2007arXiv0707.2286B}) , Mahony(\cite{Mahony}, \cite{doi:10.1002/rnc.1638}, \cite{2008arXiv0805.0828L}, \cite{5717043}), Chirkjian(\cite{Park2008}, , Maithripala(\cite{MaiBerDay2004}).
These attempts have been followed by various authors in (\cite{2015arXiv150307178W}, \cite{4282363}, \cite{7798985}, \cite{IzaSamSanKum2015}). All these efforts have been focussed on synthesizing continuous time observers for continuous time dynamical systems on Lie groups.
However, the need for discrete time observers, cannot be understated, and further, discretization of dynamical systems over manifolds is not as straightforward.

Three observers for the special orthogonal group $\SO3$ were proposed by \cite{Mahony}.
A lot of similar estimators have been developed since then, such as \cite{7442107} and \cite{2016arXiv160606208B}.
This observer is important since it is widely used, presumably with Euler discretization and the observed values converge exponentially to the desired state.
However, attempts at designing discrete observers for continuous time dynamical systems on manifolds have been distinctly missing.
The importance of such an observer cannot be understated. Finally, algorithms are discrete, and measurements too are available in today's digital world in a discrete manner.
The procedure of discretization of observers synthesized in continuous time for systems evolving on manifolds would bring in numerical inaccuracies, since special integration schemes are required to solve such equations.
Usually, continuous time systems are converted to discrete time observers using discretization techniques like Euler Discretization and Runge-Kutta Method, with a suitably small step size.
However, since these methods do not respect the constraints imposed due to the manifold structure, discretization is performed by comparatively complex discretization techniques.
However, this is tedious.
Hence, the need for developing a discrete time observer for a continuous time system on a manifold cannot be understated.

	\section{Preliminaries}
		\label{sec:lit}
		% LITERATURE REVIEW{{{
%We see that manifolds have a certain structure which makes them more difficult to work, as far as integration schemes go,
%was compared to Euclidean spaces.
%The manifold is invariant with respect to the vector field defined on that manifold and hence the state %computed by the integration scheme should always lie on the manifold.
%
%
%A few observers have been developed on manifolds.
However, since most measurements are discrete, we are interested in designing a discrete time observer for the system \pref{eq:sys}.
A contribution in this field was given by \cite{Deza}, which proposes a discrete extended Kalman filter for a given continuous system.
However, this is developed on $\R^n$.
In this article, we are interested in constructing such discrete time observers for systems on manifolds.
A few results for estimators on manifolds are now presented.
\cite{bourmaud:hal-00903252} consider a discrete system
\begin{align}
	X_k = f(X_{k-1}, u_{k-1}, n_{k-1})
\end{align}
where the noise $n_k$ is a Gaussian on the Lie group.
They use the logarithmic and exponential map to design discrete-extended Kalman filters on these manifolds, which lead to the system being restricted to manifold.
Similarly, \cite{Hauberg2013} develops an unscented Kalman filter on a Riemannian manifold, where the author uses properties of the Riemannian metric, like the exponential and the logarithmic map to calculate the predict and update state.
Although this is a superior result to \cite{bourmaud:hal-00903252} due to being a better filter for nonlinear systems, it has a restriction that the manifold needs a Riemannian metric to be defined on it.
It also requires the calculation of the logarithm map to be calculated, which is computationally expensive, since calculating the logarithm map generally requires solving an optimal control problem.

Discretization techniques such as Euler-step, Runge-Kutta have lower computational complexity
and proven robustness. We intend to use these discretization techniques
by embedding the manifold in an ambient Euclidean space. To do so, however, we first need to modify the system dynamics such that the system trajectories are always attracted to the manifold.
Such a methodology is provided by the scheme of feedback integrators.
%}}}

\subsection{Feedback Integrators} % (fold)
\label{subsec:feedback_integrators}
% FEEDBACK INTEGRATORS{{{
Given a dynamical system on a manifold, any numerical integration scheme requires us to respect the manifold structure and the first integrals of the equations of motion.
However, during conventional discretization, a lot of these quantities are violated.
For example, if our system evolves on the unit sphere $S^1$, Euler discretization will not
ensure that the trajectory stays on $S^1$.
Hence, we cannot apply normal integration schemes directly on this system.
The authors in \cite{Chang_FI}  propose a change in the system dynamics such that the original dynamics are preserved on the manifold. However, if the state is not on the manifold, the manifold itself becomes an attractor to the system, which leads to the state trajectory converging to the manifold.
The results of the paper are summarised in Theorem~\pref{theo:FI}, which is found in Appendix \pref{app:FI}.
As we can see, the modified system in Equation~\pref{eq:modified_FI_system} is an equivalent system to the system in Equation~\pref{eq:FI_system} and evolves in the ambient Euclidean space.
Since we now have a system defined on $\R^n$, we can use existing theorems in the Euclidean space to design observers.
%}}}}}}

	\section{Modified Mahony observer in Euclidean space}
		\label{sec:cont_extension}
		% EXTENSION TO AMBIENT SPACE
% SYSTEM {{{
We consider the following kinematic system evolving on the rotation group  $\SO3$,
\begin{align}\label{eq:sys}
	\dot{R} = R\Omega_{\times} \quad R \in \SO3, \Omega \in \R^3
\end{align}
with continuous time measurements given by
\begin{subequations}\label{eq:measurement_cont}
\begin{align}
	R^y &= R\label{eq:measurement_cont_R}\\
	\Omega^y &= \Omega + b\label{eq:measurement_cont_Omega}
\end{align}
\end{subequations}
where $b$ is a constant bias and the superscript $y$ denotes that the variable is a measured quantity. Although the assumptions on the measurements seem to suggest that the filter has exact information on the state $R$, this, however, is not true since they are usually corrupted by noise.

The objective is to design a continuous time observer of the continuous time system \pref{eq:sys}  with measurements \pref{eq:measurement_cont} such that
\begin{align*}
	\lim_{t \to \infty} \hat{R}(t) = R(t), \quad \hat{R} \in \R^{3\times 3}
\end{align*}
where $\hat{R}(t)$ is the estimate to $R(t) \in \SO3$, based on Euclidean integration schemes.
%}}}

% STATING MAHONY{{{
The passive observer proposed by \cite{Mahony} is given by
\begin{subequations}\label{eq:mahony_passive_filter}
\begin{align}
	\dot{\hat{R}} &= \hat{R}\left( \Omega^y - \hat{b} + k_P\omega \right)_{\times}\label{eq:mahony_passive_filter_R}\\
	\dot{\hat{b}} &= -k_I\omega\label{eq:mahony_passive_filter_b}\\
	\omega &= vex(\mathbb{P}_a(\hat{R}^TR^y))\label{eq:mahony_passive_filter_omega}
\end{align}
\end{subequations}

We now state the theorem for the convergence of the filter from \cite{Mahony}.
\begin{theorem}\label{theo:mahony}
    Consider the rotation kinematics \pref{eq:sys} and with measurements given by \pref{eq:measurement_cont}.
    Let $(\hat{R}(t), \hat{b}(t))$ denote the solution of the system \pref{eq:mahony_passive_filter}.
    Define the error variable $\tilde{R} = \hat{R}^TR$ and $\tilde{b} = b - \hat{b}$.
    Assume that $\Omega(t)$ is a bounded, absolutely continuous signal and that the pair of signals $(\Omega(t), \tilde{R})$ are asymptotically independent.
    Define $\mathbb{U}_0 \subset \SO3 \times \R^3$ by
	\begin{align}
		\mathbb{U}_0 = \left\{ (\tilde{R}, \tilde{b}) \mid \tr(\tilde{R}) = -1, \tilde{b} = 0 \right\}.
	\end{align}
	Then:
	\begin{enumerate}
		\item The set $\mathbb{U}_0$ is forward invariant and unstable with respect to the dynamic system \pref{eq:mahony_passive_filter}.
		\item The error $(\tilde{R}(t), \tilde{b}(t))$ is locally exponentially stable to $(I, 0)$.
		\item For almost all initial conditions $(\tilde{R}_0, \tilde{b}_0) \notin \mathbb{U}_0$ the trajectory $(\hat{R}(t), \hat{b}(t))$ converges to the trajectory $(R(t), b)$.
	\end{enumerate}
\end{theorem}
\hfill $\Box$

The following comments are in order:
\begin{itemize}
	\item The term $\hat{R}^TR^y$ is the error in the estimate, and the last equation constructs $\omega \in \R^3$ is based on a measure of this error.
	\item The second equation constructs an estimate, $\hat {b}$, of the bias, based on an integral term involving $\omega$.
	\item The first equation incorporates the estimate of the bias $\hat {b}$, the vector $\omega$ and the measurement $\Omega^y$ into an estimate for $R$.
\end{itemize}
However, since this observer evolves on $\SO3$, discrete implementation using Euler discretization of the observer dynamics may lead to the estimate deviating from the manifold itself. Here is where our contribution comes in.
To correct this deviation, and employ conventional Euler integrators to implement the observer, we adopt a recently proposed technique termed {\it feedback integrators} \cite{Chang_FI}.
The idea is explained in Appendix \pref{app:FI}.%}}}

% USING FI ON MAHONY{{{
The ideas discussed are now implemented for the observer structure we have.
The potential-like function that appears in \cite{Chang_FI} to be defined and added to the dynamics is of the form $V = \frac12 k_e \| \hat{R}^T\hat{R} - I \|^2$.
This additional term to the dynamics $\nabla V$ satisfies all the
required conditions as stated in \cite{Chang_FI}.

\begin{theorem}
	Consider the rotation kinematics (\ref{eq:sys}) with measurements given by (\ref{eq:measurement_cont}).
	Let $(\hat{R}(t), \hat{b}(t)) \in \R^{3\times 3}\times \R^3$ denote the solution of
	\begin{subequations}\label{eq:extended_mahony}
	\begin{align}
		\dot{\hat{R}} &= \hat{R}(\Omega^y - \hat{b} + k_p\omega)_\times - k_e\hat{R}(\hat{R}^T\hat{R} - I), & \hat{R}(0) & = \hat{R}_0\\
		\dot{\hat{b}} &= - k_I \omega, & \hat{b}(0) & = \hat{b}_0\\
		\omega &= vex(\mathbb{P}_a(\hat{R}^TR^y))
	\end{align}
	\end{subequations}
	Define the error in the estimates of $R$ and $b$ as $\tilde{R} = \hat{R}^TR$ and $\tilde{b} = b - \hat{b}$.
	Assume that $\Omega(t)$ is a bounded, absolutely continuous signal and that the pair of signals $(\Omega(t), \tilde{R})$ are asymptotically independent.
	Define $\mathbb{U}_0 \subset \SO3 \times \R^3$ by
	\begin{align*}
		\mathbb{U}_0 = \left\{ (\tilde{R}, \tilde{b}) \mid \tilde{R} \in \SO3, \tr(\tilde{R}) = -1, \tilde{b} = 0 \right\}
	\end{align*}
	Then:
	\begin{enumerate}
		\item The set $\mathbb{U}_0$ is forward invariant and unstable with respect to the dynamic system (\ref{eq:extended_mahony}).
		\item The error $(\tilde{R}(t), \tilde{b}(t))$ is locally exponentially stable to $(I, 0)$.
		\item For almost all initial conditions $(\tilde{R}_0, \tilde{b}_0) \notin \mathbb{U}_0$ the trajectory $(\hat{R}(t), \hat{b}(t))$ converges to the trajectory $(R(t), b)$.
	\end{enumerate}
\end{theorem}

\begin{proof}%{{{
	We first derive the error dynamics of the observer system.
	Differentiating $\tilde{R} = \hat{R}^TR$,
	\begin{align*}
		\dot{\tilde{R}} &= \hat{R}^T\dot{R} +\dot{\hat{R}}^TR \\
		&= [\tilde{R}, \Omega_\times] - k_P\omega_\times\tilde{R} - \tilde{b}_\times\tilde{R} - k_e(\tilde{R}\tilde{R}^T - I)\tilde{R}
	\end{align*}
	where we have used the measurements as specified in equations (\ref{eq:measurement_cont}). We also have
	\begin{align*}
		\dot{\tilde{b}} = k_I\omega
	\end{align*}

	Hence, the estimation error system is
	\begin{subequations}\label{eq:extended_mahony_error}
	\begin{align}
		\dot{\tilde{R}} &= [\tilde{R}, \Omega_\times] - k_P\omega_\times\tilde{R} - \tilde{b}_\times\tilde{R} - k_e(\tilde{R}\tilde{R}^T - I)\tilde{R}\\
		\dot{\tilde{b}} &= k_I\omega\\
		\omega &= vex(\mathbb{P}_a(\hat{R}^TR^y))
	\end{align}
	\end{subequations}

	{\bf Step 1}: We first consider the convergence of the system from $\R^{3 \times 3}$
	the ambient Euclidean space)  to $\SO3$ (the manifold).

	Consider the function
	\begin{align}
		V_1 &= \|\hat{R}^T\hat{R} - I\|^2, \quad \hat{R} \in \R^{3 \times 3}\label{eq:proof_V1_Rhat}\\
		&= \tr((\hat{R}^T\hat{R} - I)^T(\hat{R}^T\hat{R} - I)) = \tr(\hat{R}^T\hat{R}\hat{R}^T\hat{R} - 2\hat{R}^T\hat{R} + I)\nonumber
	\end{align}

	Differentiating the above function and using assumptions from \cite{Chang_FI}, we have
	\begin{align}
		\frac{dV_1}{dt} = -k_e\|\hat{R}(\hat{R}^T\hat{R} - I)\|^2\label{eq:proof_V1dot}
	\end{align}

	We note that $V_1$ can also be written as
	\begin{align}
		V_1 = \|\tilde{R}\tilde{R}^T - I\|^2, \quad \tilde{R} \in \R^{3 \times 3}\label{eq:proof_V1_Rtilde}
	\end{align}

	Notice that $\|\hat{R}(\hat{R}^T\hat{R} - I)\| > 0 \, ~ \forall ~ \hat{R} \notin \SO3$, hence the derivative will be negative whenever $\hat{R} \notin \SO3$.
	This implies that given an $\epsilon > 0$, $\exists~ T > 0$, such that for all $t > T$,
	\begin{align}
		&\|\hat{R}^T(t)\hat{R}(t) - I\|^2 < \epsilon\\
		\Rightarrow ~ &\|\tilde{R}(t)\tilde{R}^T(t) - I\|^2 < \epsilon\label{eq:norm_RRI_limit}
	\end{align}

	{\bf Step 2}: Before proving the convergence of the observed state to the actual
	state, we prove some intermediate results that are used later in the proof.

	Define the inner product between two elements of $\R^{3\times 3}$ as
	\begin{align*}
		\langle A,B \rangle &\deff \tr(A^TB)\\
		\Rightarrow \|A\| &= \sqrt{\tr(A^TA)}
	\end{align*}

	Using the Cauchy Schwarz Inequality on this inner product, we have
	\begin{align*}
		& |\tr(\tilde{R}^T(t)\tilde{R}(t) - I)| = \langle\tilde{R}^T(t)\tilde{R}(t) - I, I\rangle \leq \|I\|\|\tilde{R}^T(t)\tilde{R}(t) - I\|\\
		\Rightarrow ~& |\tr(\tilde{R}^T(t)\tilde{R}(t) - I)| \leq \sqrt{3\epsilon} \Rightarrow ~ 3 - \sqrt{3\epsilon} \leq \tr(\tilde{R}^T(t)\tilde{R}(t)) \leq 3 + \sqrt{3\epsilon}\\
		\Rightarrow ~& \|\tilde{R}(t)\| \leq \sqrt{3 + \sqrt{3\epsilon}} \Rightarrow ~ \|\tilde{R}(t)\| \leq \sqrt{3}\numberthis\label{eq:norm_R_limit}
	\end{align*}
	which is valid for small $\epsilon $.
	Again using the Cauchy Schwarz Inequality, we have,
	\begin{align}
		|\tr\left((\tilde{R}\tilde{R}^T - I)\tilde{R}\right)| \leq \|\tilde{R}\tilde{R}^T - I\|\|\tilde{R}\| \Rightarrow |\tr\left((\tilde{R}\tilde{R}^T - I)\tilde{R}\right)| \leq \sqrt{\epsilon}\sqrt{3} = \sqrt{3\epsilon}\label{eq:tr_R_RRI_limit}
	\end{align}

	{\bf Step 3}: We now prove the convergence of the observer to the true values.

	To prove that the estimates computed by (\ref{eq:extended_mahony}) converge to the true values, we choose the Lyapunov function
	\begin{align}
		V = \frac14 \|I_3 - \tilde{R}\|^2 + \frac{1}{2k_I}\|\tilde{b}\|^2
	\end{align}

	Using the fact that $\|A\|^2 = \tr(A^TA)$, we differentiate the above equation to obtain,
	\begin{align*}
		\frac{dV}{dt} &= \frac12 \tr\left( (\tilde{R}^T - I)\frac{d\tilde{R}}{dt} \right) + \frac{1}{k_I}\tilde{b}^T\frac{d\tilde{b}}{dt}\nonumber\\
		&= \frac12 \tr\left((\tilde{R}^T - I)([\tilde{R}, \Omega_\times] - k_P\omega_\times\tilde{R} - \tilde{b}_\times\tilde{R} - k_e(\tilde{R}\tilde{R}^T - I)\tilde{R})\right) + \frac{1}{k_I}\tilde{b}^T\frac{d\tilde{b}}{dt}\nonumber\\
		&= \frac12 \tr\left(([\tilde{R}, \Omega_\times] - k_P\omega_\times\tilde{R} - \tilde{b}_\times\tilde{R} - k_e(\tilde{R}\tilde{R}^T - I)\tilde{R})\tilde{R}^T\right)\nonumber \\&\quad\ -\frac12\tr\left([\tilde{R}, \Omega_\times] - k_P\omega_\times\tilde{R} - \tilde{b}_\times\tilde{R} - k_e(\tilde{R}\tilde{R}^T - I)\tilde{R}\right) + \frac{1}{k_I}\tilde{b}^T\frac{d\tilde{b}}{dt}\\
		&= \frac12 \tr\left(\tilde{R}\Omega_\times\tilde{R}^T - \Omega_\times\tilde{R}\tilde{R}^T - k_P\omega_\times\tilde{R}\tilde{R}^T - \tilde{b}_\times\tilde{R}\tilde{R}^T - k_e(\tilde{R}\tilde{R}^T - I)\tilde{R}\tilde{R}^T\right)\nonumber \\&\quad -\frac12\left(\tr\left([\tilde{R}, \Omega_\times]\right) - k_P\tr\left(\omega_\times\tilde{R}\right) - \tr\left(\tilde{b}_\times\tilde{R}\right) - k_e\tr\left((\tilde{R}\tilde{R}^T - I)\tilde{R}\right)\right) +  \frac{1}{k_I}\tilde{b}^T\frac{d\tilde{b}}{dt}\\
		&= \frac12 \tr\left(\tilde{R}\Omega_\times\tilde{R}^T\right) -\frac12 \tr\left(\Omega_\times\tilde{R}\tilde{R}^T\right) - \frac12\tr\left(k_P\omega_\times\tilde{R}\tilde{R}^T\right) - \frac12\tr\left(\tilde{b}_\times\tilde{R}\tilde{R}^T\right)\nonumber \\&\quad - \frac12\tr\left(k_e(\tilde{R}\tilde{R}^T - I)\tilde{R}\tilde{R}^T\right) - k_P\langle\omega, vex(\mathbb{P}_a(\tilde{R}))\rangle - \langle\tilde{b}, vex(\mathbb{P}_a(\tilde{R}))\rangle \nonumber \\&\quad + \frac{k_e}{2}\tr\left((\tilde{R}\tilde{R}^T - I)\tilde{R}\right) -  \frac{1}{k_I}\langle\tilde{b},\dot{\hat{b}}\rangle\\
		&= \frac12 \tr\left(\Omega_\times\tilde{R}^T\tilde{R}\right) -\frac12 \tr\left(\Omega_\times\tilde{R}\tilde{R}^T\right) - \frac12\tr\left(k_P\omega_\times\tilde{R}\tilde{R}^T\right) - \frac12\tr\left(\tilde{b}_\times\tilde{R}\tilde{R}^T\right)\nonumber \\&\quad - \frac12\tr\left(k_e(\tilde{R}\tilde{R}^T - I)\tilde{R}\tilde{R}^T\right) +  \frac{k_e}{2}\tr\left((\tilde{R}\tilde{R}^T - I)\tilde{R}\right) - k_P\|\omega\|^2\\
		&= - \frac{k_e}{2}\left(\tr\left((\tilde{R}\tilde{R}^T - I)^2\right) + \tr\left(\tilde{R}\tilde{R}^T - I\right)\right) +  \frac{k_e}{2}\tr\left((\tilde{R}\tilde{R}^T - I)\tilde{R}\right) - k_P\|\omega\|^2\\
		&= - \frac{k_e}{2}\left(\tr\left((\tilde{R}\tilde{R}^T - I)^T(\tilde{R}\tilde{R}^T - I)\right) + \tr\left(\tilde{R}\tilde{R}^T - I\right)\right) +  \frac{k_e}{2}\tr\left((\tilde{R}\tilde{R}^T - I)\tilde{R}\right ) \nonumber \\&\quad- k_P\|\omega\|^2\\
		&= - \frac{k_e}{2}\left(\|\tilde{R}\tilde{R}^T - I\|^2 + \tr\left(\tilde{R}\tilde{R}^T - I\right)\right) +  \frac{k_e}{2}\tr\left((\tilde{R}\tilde{R}^T - I)\tilde{R}\right) - k_P\|\omega\|^2\numberthis\label{eq:dvdt_extended_mahony}
	\end{align*}

	Substituting (\ref{eq:norm_RRI_limit}), (\ref{eq:norm_R_limit}) and (\ref{eq:tr_R_RRI_limit}) in (\ref{eq:dvdt_extended_mahony}), we have
	\begin{align*}
		-\frac{k_e}{2}\sqrt{3\epsilon} - \frac{k_e}{2}\epsilon - \frac{k_e}{2}\sqrt{3\epsilon} - k_P\|\omega\|^2 \leq \frac{dV}{dt} \leq & \frac{k_e}{2}\sqrt{3\epsilon} - \frac{k_e}{2}\epsilon + \frac{k_e}{2}\sqrt{3\epsilon} - k_P\|\omega\|^2\\
		\Rightarrow \frac{dV}{dt} \leq k_e\sqrt{3\epsilon} - \frac{k_e}{2}\epsilon - k_P\|\omega\|^2\label{eq:dvdt_extended_mahony_epsilon}\numberthis
	\end{align*}
	The term $\| \omega \| = \| vex (\mathbb{P}_a(\hat{R}^TR^y)) \| $ is non-zero due to the error between the estimate and the measured value of the system being non-zero. Hence, for reasonably small $\epsilon > 0$  and a suitably large value of $k_P$, the right hand side of the inequality is negative.Hence, the system is asymptotically stable.
	Since $\epsilon$ can be chosen arbitrarily small, we have that $V \to 0$.

	The rest of the proof follows on the same lines as \cite{Mahony}. The earlier part has shown that no invariant set of the error dynamics (\ref{eq:extended_mahony_error}) can lie in $\R^{3\times 3} \setminus \SO3$ due to the system dynamics asymptotically converging to the manifold. Hence, any invariant set of the error dynamics must lie in $\SO3$.
	Hence, the equilibrium points or limit sets of the system, if any, lie completely in $\SO3$.
	Since $\tilde{R} \in \SO3 \Rightarrow \exists~ P, B$ such that $\tilde{R} = PBP^{-1}$ where
	\begin{align*}
		B = \begin{bmatrix}
			1 & 0 & 0\\
			0 & \cos\theta & \sin\theta\\
			0 & -\sin\theta & \cos\theta
		\end{bmatrix}
	\end{align*}

	Since $\tilde{R} \in \SO3$, and that $\|\mathbb{P}_a(\tilde{R})\| = \sqrt{2}\sin\theta$, we have $\omega = 0$ implies either $\theta = 0$ or $\theta = \pi$. $\theta = 0$ implies $\tilde{R} = I$ and $\theta = \pi$ implies $\tr(\tilde{R}) = -1$.
	If $\tilde{R} = I$, we have $\tilde{b} = 0$ as the equilibrium point from the system dynamics.

	The terms $\tr(\tilde{R}) = -1, \tilde{R} \in \SO3$ reduce the error dynamics (\ref{eq:extended_mahony_error}) to
	\begin{subequations}\label{eq:extended_mahony_reduced_U0}
	\begin{align}
		\dot{\tilde{R}} & = [\tilde{R}, \Omega_\times] - \tilde{b}_\times\tilde{R} &
		\dot{\tilde{b}}  = 0
	\end{align}
	\end{subequations}

	Differentiating $\mathbb{P}_a(\tilde{R}) = 0$, we get $\mathbb{P}_a(\tilde{b}_\times\tilde{R}) = 0$.
	Let $\mathbb{U} = \{(\tilde{R}, \tilde{b}) \mid \tilde{R} \in \SO3, \tr(\tilde{R}) = -1, \mathbb{P}_a(\tilde{b}_\times\tilde{R}) = 0\}$.

	We prove by contradiction that $\mathbb{U}_0 \subset \mathbb{U}$ is the largest forward invariant set of the closed loop dynamics \ref{eq:extended_mahony_error}.
	The solution for the reduced error dynamics is
	\begin{align}
		\tilde{R}(t) = \exp\left(-\int_0^t\Omega_\times \textrm{d}t\right)\tilde{R}_0\exp\left(\int_0^t\Omega_\times \textrm{d}t\right)
	\end{align}

	where $\tilde{R}_0 = \tilde{R}(0) \in \mathbb{U}_0$. Since $\tilde{R}_0 \in \SO3$, $\tilde{R} \in \SO3$ due to orthogonality of $\exp(\int_0^t\Omega_\times \dt)$.
	Hence, $\mathbb{U}_0$ is forward invariant.

	Assume that there exists $(\tilde{R}_0, \tilde{b}_0) \in \mathbb{U} - \mathbb{U}_0$ such that $(\tilde{R}(t), \tilde{b}(t)) \in \mathbb{U} ~ \forall ~ t > 0$.
	We have $\mathbb{P}_a(\tilde{b}_\times\tilde{R}) = 0$ on this trajectory. Differentiating,
	\begin{align*}
		\frac{d}{dt}\mathbb{P}_a(\tilde{b}_\times\tilde{R}) & = \mathbb{P}_a(\tilde{b}_\times[\tilde{R}, \Omega_\times]) - \mathbb{P}_a(\tilde{b}_\times\tilde{b}_\times\tilde{R})\\
		& = \mathbb{P}_a(\tilde{b}_\times[\tilde{R}, \Omega_\times]) - \mathbb{P}_a(\tilde{b}_\times\tilde{R}\tilde{b}^T_\times)\\
		& = \mathbb{P}_a(\tilde{b}_\times[\tilde{R}, \Omega_\times])\\
		& = -\frac12 \left((\tilde{b}\times\Omega)_\times\tilde{R} + \tilde{R}(\tilde{b}\times\Omega)_\times\right) = 0\numberthis\label{eq:proof_U_0_largest}
	\end{align*}
	where we use the fact that $\mathbb{P}_a(\tilde{b}_\times\tilde{R}) = 0$.
	Since $(\Omega(t), \tilde{R}(t))$ are asymptotically independent, the equation (\ref{eq:proof_U_0_largest}) must be degenerate.
	This implies that there exists a time $T$ such that for all $t > T, \tilde{b}(t) = 0$, which implies that $\tilde{b} = 0 ~ \forall ~ t > 0$. This contradicts the assumption.
	Hence, $\mathbb{U}_0$ is the largest invariant set of the dynamics.

	{\it Local exponential convergence}: To prove local exponential convergence, consider the linearisation of the system dynamics about $(I, 0)$.
	Assume
	\begin{align*}
		\tilde{R} &= I + s + a_\times\\
		\tilde{b} &= -y
	\end{align*}
	where s is symmetric. This yields the linearisation of the error dynamics (\ref{eq:extended_mahony_error}) as
	\begin{subequations}\label{eq:extended_mahony_linear}
	\begin{align}
		\frac{d}{dt} \begin{pmatrix} a\\ y\end{pmatrix} &= \begin{pmatrix}
			-k_pI - \Omega(t)_\times & I\\
			-k_II & 0
		\end{pmatrix} \begin{pmatrix} a\\ y\end{pmatrix}\label{eq:ext_mah_err_lin_mah}\\
		\dot{s} &= [s, \Omega_\times] - 2k_es
	\end{align}
	\end{subequations}

	To show local exponential convergence of $s$ to $0$, we substitute the linearisation in equations (\ref{eq:proof_V1_Rtilde}) and (\ref{eq:proof_V1dot}) noting the fact that $\hat{R}^T\hat{R} = \tilde{R}\tilde{R}^T$ and
	\begin{align*}
		V_1 & = \| s \|^2\\
		\frac{d}{dt} V_1 & = -k_e\|s\|^2
	\end{align*}
	Hence, $V_1 \to 0$ exponentially and $s \to 0$ exponentially.
	For proof of convergence of $(a,y)$ to $(0,0)$, let $|\Omega_{max}|$ denote the maximum value attainable by $\Omega$ and choose
	\begin{align*}
		\alpha_2 > 0, \quad \alpha_1 > \frac{\alpha_2(|\Omega_{max}|^2 + k_I)}{k_p}, \\
		\frac{\alpha_1 + k_p\alpha_2}{k_I} < \alpha_3 < \frac{\alpha_1 + k_p\alpha_2}{k_I} + \frac{|\Omega_{max}|\alpha_2}{k_I}
	\end{align*}
	such that the matrices
	\begin{align*}
		P = \begin{pmatrix} \alpha_1 I & \alpha_2 I \\ -\alpha_2 I & \alpha_3 I \end{pmatrix}, \;
		Q = \begin{pmatrix} k_p\alpha_1 - \alpha_2 k_I & \alpha_2 |\Omega_{max}| \\ -\alpha_2 |\Omega_{max}| & \alpha_2 \end{pmatrix}
	\end{align*}
	are positive definite.
	Consider the cost function $W = \frac12 \xi^T P \xi$, with $\xi = (a,y)^T$.
	Differentiating $W$,
	\begin{align*}
		\dot{W} = -(k_p \alpha_1 - \alpha_2 k_I)\|a\|^2 - \alpha_2\|y\|^2 + y^Ta(\alpha_1 + k_p \alpha_2 - \alpha_3 k_I) + \alpha_2y^T(\Omega \times a)
	\end{align*}
	which leads to
	\begin{align*}
		\frac{d}{dt}\left(\xi^T P \xi\right) \leq -2(\|a\|, \|y\|) Q \begin{pmatrix} \|a\| \\ \|y\| \end{pmatrix}
	\end{align*}

	Hence, the observer system is locally exponentially stable.

\end{proof}%}}}

	\section{Discretised Observer}
		\label{sec:disc_obs}
		% DISCRETIZATION
We now consider the same kinematic system \pref{eq:sys} as before but with discrete measurements, and the measured variables being
\begin{subequations}\label{eq:measurement_disrete}
\begin{align}
	R^y_k &= R(k\dt)\\
	\Omega^y_k &= \Omega(k\dt) + b
\end{align}
\end{subequations}
where $\dt$ is the discretization step size.
The objective is to design a discrete time observer of the continuous time system \pref{eq:sys}  with measurements \pref{eq:measurement_disrete} such that
\begin{align}
	\lim_{k \to \infty} \| \hat{R}_d(k \mid k) - R(k\dt) \| < \epsilon (\dt)
\end{align}
where the order of the error, $\epsilon$ is dependant on $\dt$.
Here, $\hat{R}_d(k \mid k)$ is the discrete time estimate of the state at the $k$th instant based
on information till the $k$th instant of time.
Before proceeding further, we first define the term {\it convergent} in the case of discrete integrators
\begin{definition}\cite{iserles2009first}
	An integrator method is said to be convergent if, for every O.D.E.
	\begin{align*}
		y' = f(t, y), \; t \geq t_0, \; y(t_0) = y_0
	\end{align*}
	with a Lipschitz function $f$ and for every $t^* > 0$ , the following equality
		\begin{align*}
		\lim_{h\to 0+} \max_{n = 0, 1, \ldots, \lfloor t^*/h \rfloor} \| y_{n,h} - y(t_n) \| = 0
	\end{align*}
	where $y_{n,h}$ is the numerical estimate of $y$ after $n$ steps, each with step size $h$, holds.
\end{definition}

We now propose a two-step discrete time observer:
\begin{theorem}
	Consider the rotational kinematics given by \pref{eq:sys} with measurements given by \pref{eq:measurement_disrete}. Assume that $\Omega(t)$ is bounded.
	Let $(\hat{R}_d(k \mid k), \hat{b}_d(k))$ denote the observed state at stage $k$ based on information till stage $k$. Now consider the predictor-corrector system given by

	\begin{subequations}\label{eq:discrete_observer}
	\underline{Predictor step}:  $t \in [(k-1)\dt, k \dt[$
	\begin{align}
		\hat{R}_d(k\mid k-1) = \hat{R}_d(k-1\mid k-1)\exp(\hat{\Omega}_d(k-1)_\times\dt), \quad \hat{R}_d(0 \mid 0) = \hat{R}_{d,0}
	\end{align}

	\underline{Corrector step}: at $t = k\dt$
	\begin{align}
		\omega_k &= vex(\mathbb{P}_a(\hat{R}_d(k\mid k-1)^TR^y_k))\label{eq:eR}\\
		\begin{split}
			\hat{R}_d(k\mid k) &= \hat{R}_d(k\mid k-1) + \hat{R}_d(k\mid k-1)k_p\omega_{k_\times}\dt \\&\quad- k_e\hat{R}_d(k\mid k-1)(\hat{R}_d(k\mid k-1)^T\hat{R}_d(k\mid k-1) - I)\dt \label{eq:Rupd}
		\end{split}\\
		\hat{b}_d(k) &= \hat{b}_d(k-1) + k_b \omega_k\dt\label{eq:bupd}, \quad \quad \quad \quad \quad \quad \quad \quad \quad \quad \hat{b}_d(0) = \hat{b}_{0,d}\\
		\hat{\Omega}_d(k) &= \Omega_k^y - \hat{b}_d(k)\label{eq:omupd}
	\end{align}
	\end{subequations}
	Given a sufficiently small $\dt > 0 ~ \exists ~ M(\dt) \in \mathbb{Z}^+$ and $\epsilon(\dt) > 0$ such that
	\begin{align}
		\| \hat{R}_d(k\mid k) - R(k\dt) \| < \epsilon ~ \forall ~ k > M
	\end{align}
	Moreover, for a sufficiently small $\dt$, the state $(\hat{R}_d(k \mid k), \hat{b}_d(k))$ is locally exponential stable to $(I, 0)$.
\end{theorem}

\begin{proof}%{{{
	For the proposed observer, we have
	\begin{align*}
		\hat{R}_d(k\mid k-1) = \hat{R}_d(k-1\mid k-1)\exp(\hat{\Omega}_d(k-1)_\times\dt)
	\end{align*}

	From here onward, for readability, we adopt the following notation:
	\begin{align*}
		\hat{R}_{a,b} \deff \hat{R}_d(a\mid b)  \quad e^{\hat{\Omega}_{d_{(k-1)}\times}\dt} \deff \exp(\hat{\Omega}_d(k-1)_\times\dt)
	\end{align*}

	From \pref{eq:eR}, and using the above notation, we have
	\begin{align}
	    \omega_{k_\times} &= \frac{e^{-\hat{\Omega}_{d_{(k-1)}\times}\dt}\hat{R}_{k-1,k-1}^TR^y_k - (R^y_k)^T\hat{R}_{k-1,k-1}e^{\hat{\Omega}_{d_{(k-1)}\times}\dt}}{2}\label{eq:proof_discrete_omega}
	\end{align}
	and
	\begin{align*}
		\hat{R}_d(k\mid k-1)^T\hat{R}_d(k\mid k-1) - I &= \left(e^{-\hat{\Omega}_{d_{(k-1)}\times}\dt}\hat{R}_{k-1,k-1}^T\hat{R}_{k-1,k-1}e^{\hat{\Omega}_{d_{(k-1)}\times}\dt} - I\right)\\
		&= e^{-\hat{\Omega}_{d_{(k-1)}\times}\dt}(\hat{R}_{k-1,k-1}^T\hat{R}_{k-1,k-1} - I)e^{\hat{\Omega}_{d_{(k-1)}\times}\dt}\numberthis\label{eq:proof_discrete_RRI}
	\end{align*}

	Substituting \pref{eq:proof_discrete_omega} and \pref{eq:proof_discrete_RRI} in \pref{eq:Rupd}, the state of the proposed observer at $t = k\dt$ given $k$ measurements is given by
	\begin{align*}
		\hat{R}_{k,k} =& ~\hat{R}_{k-1,k-1}e^{\hat{\Omega}_{d_{(k-1)}\times}\dt}\bigg[I - k_e\dt e^{-\hat{\Omega}_{d_{(k-1)}\times}\dt}(\hat{R}_{k-1,k-1}^T\hat{R}_{k-1,k-1} - I)e^{\hat{\Omega}_{d_{(k-1)}\times}\dt}\nonumber
		\\& \quad  + k_p\dt\frac{e^{-\hat{\Omega}_{d_{(k-1)}\times}\dt}\hat{R}_{k-1,k-1}^TR^y_k - (R^y_k)^T\hat{R}_{k-1,k-1}e^{\hat{\Omega}_{d_{(k-1)}\times}\dt}}{2} \bigg]
	\end{align*}

	Now we substitute the definition of the exponential of a matrix,
	\begin{align*}
		\exp(A) = I + A + \frac{A^2}{2!} + \ldots
	\end{align*}

	and rewrite \pref{eq:proof_discrete_omega} as
	\begin{align*}
	    2\omega_{k_\times} &= e^{-\hat{\Omega}_{d_{(k-1)}\times}\dt}\hat{R}_{k-1,k-1}^TR^y_k - (R^y_k)^T\hat{R}_{k-1,k-1}e^{\hat{\Omega}_{d_{(k-1)}\times}\dt}\\
		&= [I + \hat{\Omega}(k-1)_\times\dt + \mathcal{O}(\dt^2)]\hat{R}_{k-1,k-1}^TR^y_k\nonumber \\&- (R^y_k)^T\hat{R}_{k-1,k-1}[I + \hat{\Omega}(k-1)_\times\dt + o(\dt^2)]\\
		&= \hat{R}_{k-1,k-1}^TR^y_k - (R^y_k)^T\hat{R}_{k-1,k-1} + \mathcal{O}(\dt)
	\end{align*}

	The update is given by
	\begin{align*}
		\hat{R}(k\mid k) &= \hat{R}_{k-1,k-1}[I + \hat{\Omega}(k-1)_\times\dt + \mathcal{O}(\dt^2)]\left[I\right. \nonumber
		\\&\; \left.+ k_p\dt\left(\frac{\hat{R}_{k-1,k-1}^TR^y_k - (R^y_k)^T\hat{R}_{k-1,k-1} + \mathcal{O}(\dt)}{2}\right)\right. \nonumber
		\\&\; \left.- k_e\dt[I + \hat{\Omega}(k-1)_\times\dt + \mathcal{O}(\dt^2)](\hat{R}_{k-1,k-1}^T\hat{R}_{k-1,k-1} \right.\nonumber
		\\&\; \left.- I)[I + \hat{\Omega}(k-1)_\times\dt + \mathcal{O}(\dt^2)]\right]\\
		&= \hat{R}_{k-1,k-1}[I + \hat{\Omega}(k-1)_\times\dt + \mathcal{O}(\dt^2)]\bigg[I - k_e\dt[\hat{R}_{k-1,k-1}^T\hat{R}_{k-1,k-1} - I]\nonumber
		\\&\; + k_p\dt\left(\frac{\hat{R}_{k-1,k-1}^TR^y_k - (R^y_k)^T\hat{R}_{k-1,k-1}}{2}\right) + \mathcal{O}(\dt^2)\bigg]\\
		&= \hat{R}_{k-1,k-1}\bigg[I + \hat{\Omega}(k-1)_\times\dt - k_e\dt[\hat{R}_{k-1,k-1}^T\hat{R}_{k-1,k-1} - I]\nonumber
		\\&\; + k_p\dt\left(\frac{\hat{R}_{k-1,k-1}^TR^y_k - (R^y_k)^T\hat{R}_{k-1,k-1}}{2}\right) + \mathcal{O}(\dt^2)\bigg]\\
		&= \hat{R}_{k-1,k-1}\bigg[I + \hat{\Omega}(k-1)_\times\dt - k_e\dt[\hat{R}_{k-1,k-1}^T\hat{R}_{k-1,k-1} - I]\nonumber
		\\&\; + k_p\dt\left(\frac{\hat{R}_{k-1,k-1}^TR^y_k - (R^y_k)^T\hat{R}_{k-1,k-1}}{2}\right)\bigg] + \mathcal{O}(\dt^2)\label{eq:diff_my}\numberthis
	\end{align*}
     The next part of our proof follows on similar lines as the one given in \cite{iserles2009first} demonstrating convergence of the Euler discretization. For notation purposes, we denote
	\begin{align*}
		f(t,\hat{R}) \deff \hat{R}(t)(\Omega^y_k - \hat{b}(t) - k_p\omega(t))_\times - k_e\hat{R}(t)(\hat{R}(t)^T\hat{R}(t) - I)
	\end{align*}

	Note that $f(t, \hat{R})$ is the RHS of the extended continuous time observer, which is continuously differentiable, and hence Lipschitz. Hence, $ \exists ~ L > 0$ such that
	\begin{align}
		\| f(t, x) - f(t, y) \| \leq L \| x - y \|, \quad x, y \in \R^{3\times 3}\label{eq:proof_lipschitz}
	\end{align}

	We now define the error between the extended continuous time observed value and the one from the proposed observer at time $k \dt$ as
	\begin{align*}
		e_{k,\dt} = \hat{R}_d(k \mid k) - \hat{R}(k\dt)
	\end{align*}

	Expanding in a Taylor series, we have
	\begin{align*}
		& \hat{R}(k\dt + \dt) = \hat{R}(k\dt) + \dt \dot{\hat{R}}(k\dt) + \mathcal{O}(\dt^2)\\
		\Rightarrow & \hat{R}((k + 1)\dt) = \hat{R}(k\dt) + \dt f(k\dt, \hat{R}(k\dt)) + \mathcal{O}(\dt^2) \label{eq_proof_discrete_taylor}\numberthis
	\end{align*}

	Since $\hat{R}$ is continuously differentiable and $\Omega$ bounded, $\mathcal{O}(\dt^2)$ is bounded by a term $c\dt^2$ where $c > 0$ is a constant.
	Subtracting \pref{eq_proof_discrete_taylor} from \pref{eq:diff_my} we have
	\begin{align*}
		e_{k+1, \dt} = e_{k, \dt} + \dt \left[f(k\dt, \hat{R}(k\dt) + e_{k, \dt}) - f(k\dt, \hat{R}(k\dt))\right] + \mathcal{O}(\dt^2)
	\end{align*}

	It follows from the triangle inequality that
	\begin{align*}
		\| e_{k+1, \dt} \| & \leq \| e_{k, \dt} \| + \dt \| f(k\dt, \hat{R}(k\dt) + e_{k, \dt}) - f(k\dt, \hat{R}(k\dt)) \| + c\dt^2\\
		\Rightarrow\| e_{k+1, \dt} \| & \leq (1 + L\dt)\| e_{k, \dt} \| + c\dt^2\numberthis\label{eq_proof_discrete_ek1_ek}
	\end{align*}

	We now claim
	\begin{align}
		\| e_{k, \dt} \| \leq \frac{c}{L}\dt[(1 + L\dt)^k - 1], \quad k = 0, 1, 2, \ldots\label{eq_proof_discrete_ek_explicit}
	\end{align}

	{\it Proof by induction}: \\
	Clearly, since the value at the initial time epoch is assumed to be same as that of the continuous time, $e_{0,k} = 0$.
	For general $k > 0$, we assume that \pref{eq_proof_discrete_ek_explicit} is true up to $k$ and use \pref{eq_proof_discrete_ek1_ek}
	\begin{align*}
		\| e_{k+1, \dt} \| \leq (1 + L\dt)\frac{c}{L}\dt((1 + L\dt)^k - 1) + c\dt^2 = \frac{c}{L}\dt[(1 + L\dt)^{k+1} - 1]
	\end{align*}

	Hence, \pref{eq_proof_discrete_ek_explicit} holds true.

	The constant $L\dt$ is positive, hence $(1 + L\dt) < \exp(L\dt) \Rightarrow (1+L\dt)^n < \exp(nL\dt)$.
	The index $k$ is allowed in the range $\{0, 1, \ldots, \lfloor t^*/\dt \rfloor\}$, hence $(1 + L\dt)^k < \exp(\lfloor t^*/\dt \rfloor L\dt) \newline \leq \exp(t^*L)$.

	Substituting in \pref{eq_proof_discrete_ek_explicit}, we obtain that
	\begin{align*}
		\| e_{k, \dt} \| \leq \frac{c}{L}\dt[\exp(t^*L) - 1], \quad k = 0, 1, 2, \ldots, \lfloor t^*/h \rfloor
	\end{align*}

	Since $c(\exp(t^*L) - 1)/L$ is independent of $\dt$, it follows that
	\begin{align*}
		{\lim_{\dt \to 0}} \max_{0 \leq k\dt \leq t^*} \| e_{k,\dt} \| = 0
	\end{align*}
	Since this holds for every $t^* > 0$, we have that the discrete time observer is convergent to the modified continuous time Mahony observer.
	Since the continuous time observer has an exponential convergence rate, we can conclude that the proposed predictor-corrector observer also has an exponential convergence rate given a choice of a small enough $\dt$.
\end{proof}%}}}

\begin{comment1}
The reason for such a discretization is as follows.
The predictor part of the algorithm provides an estimate based on just the rotational kinematics, assuming a fixed value of $\Omega$.
The corrector then brings in a correction term based on the measurement and on similar lines as the modified Mahony estimator.

\end{comment1}

	\section{Simulations}
		\label{sec:simulations}
		% SIMULATIONS
A few numerical experiments are performed for checking the efficacy of the observer.
We consider the response of the system with $\Omega = 1$ and the initial estimate different from the actual value of the system.
For all simulations performed, $\dt$ is taken to be \SI{0.5}{\sec} unless specified.

\subsection{Response of the system for constant $\Omega$}
% CONSTANT INPUT OMEGA RESPONSE{{{
The input angular velocity of the system is kept constant at $[1, 1, 1]$ rad/s $\forall t \geq 0$.
This simplistic value was taken to prevent issues such as aliasing and sampling time.

\iftoggle{arxiv}{
Figures (\ref{fig:err_ke}, \ref{fig:err_kp}, \ref{fig:err_kI}) show variation of the system's
performance for various values of $k_p$, $k_I$ and $k_e$.
A few conclusions can be drawn from this:
\begin{itemize}
    \item Changing only the value of $k_e$ may not change the rate of convergence of the error.
    This can be seen in Figure (\ref{fig:err_ke}).
    \item Increasing value of $k_p$ reduces the time taken by the estimate to reach the desired value and results in oscillations.
    However, increasing $k_p$ beyond a certain value also leads to the observer never reaching the desired state value and oscillating about a point with a finite error.
    Refer Figure (\ref{fig:err_kp}).
    \item Increasing value of $k_I$ increases the frequency of oscillations, but does not affect the time required to settle to the required value.
Refer Figure (\ref{fig:err_kI}).
\end{itemize}

\begin{figure}[h]
    \centering
    \includegraphics[width=0.8\linewidth]{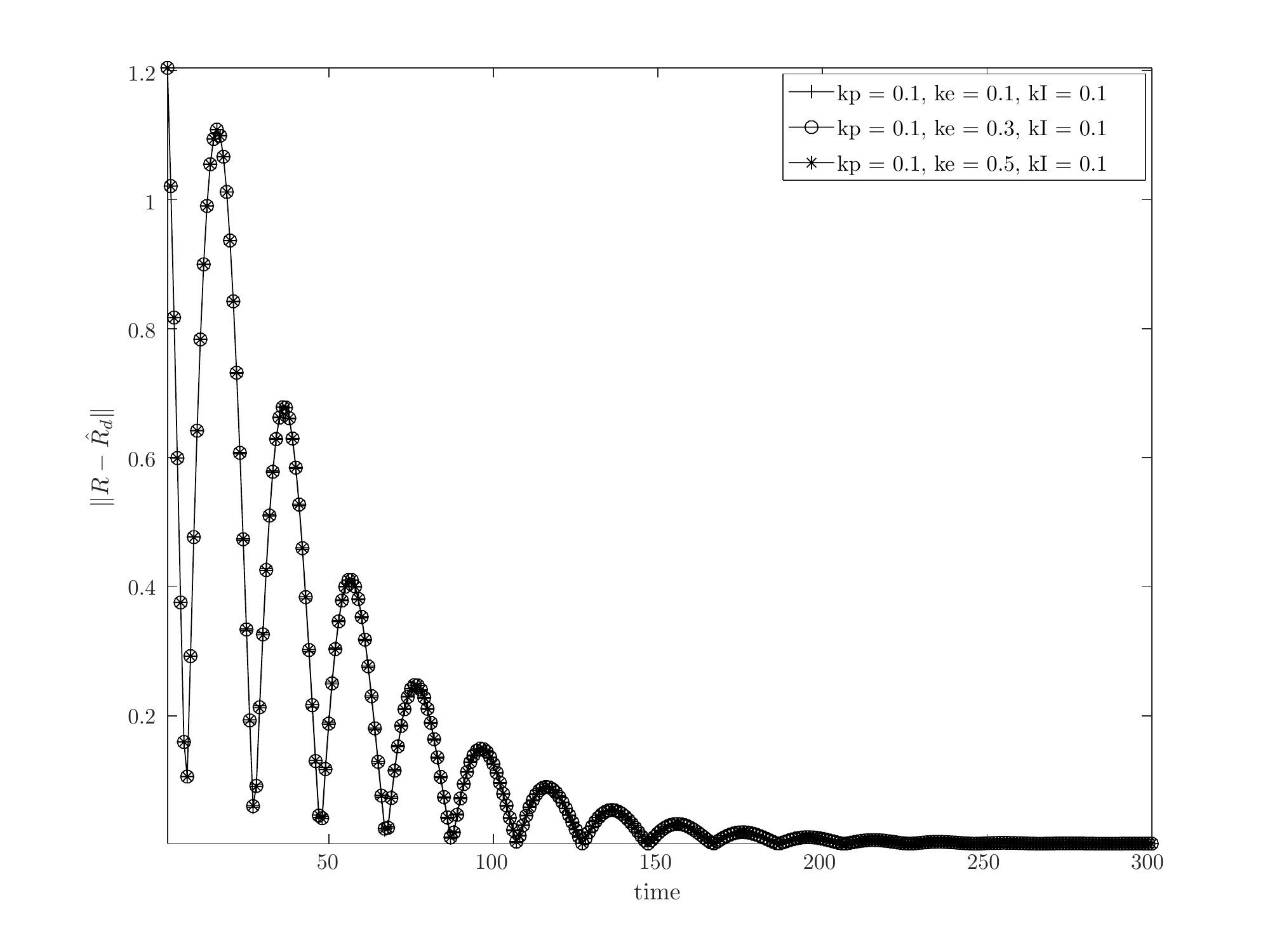}
	\caption{\label{fig:err_ke}Variation of rate of convergence with $k_e$ ($k_p$ and $k_I$ constant) ($\dt = 0.5s$)}
\end{figure}
\begin{figure}[h]
    \centering
    \includegraphics[width=0.8\linewidth]{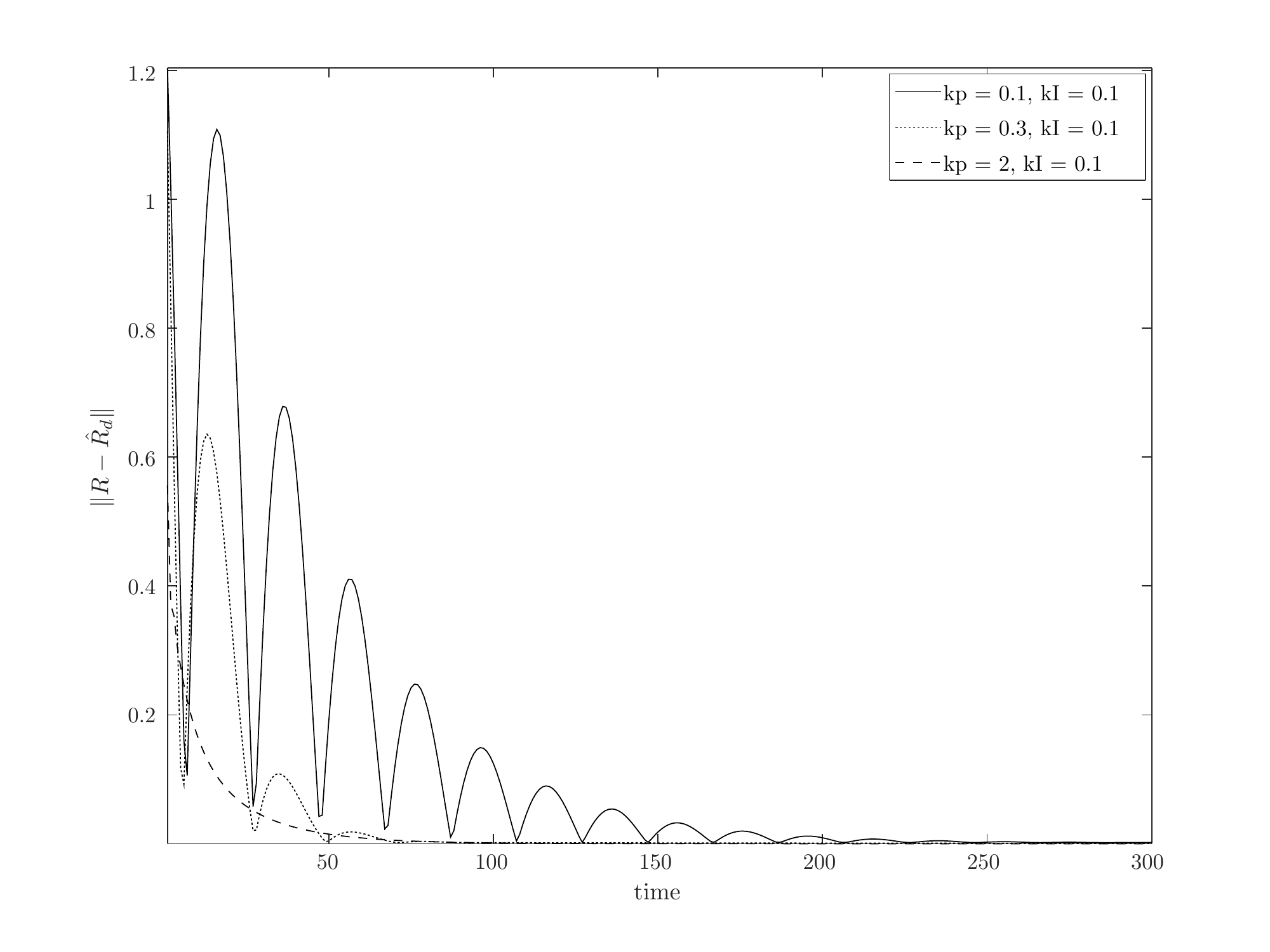}
	\caption{\label{fig:err_kp}Variation of rate of convergence with $k_p$ ($k_e$ and $k_I$ constant) ($\dt = 0.5s$)}
\end{figure}
\begin{figure}[h]
    \centering
    \includegraphics[width=0.8\linewidth]{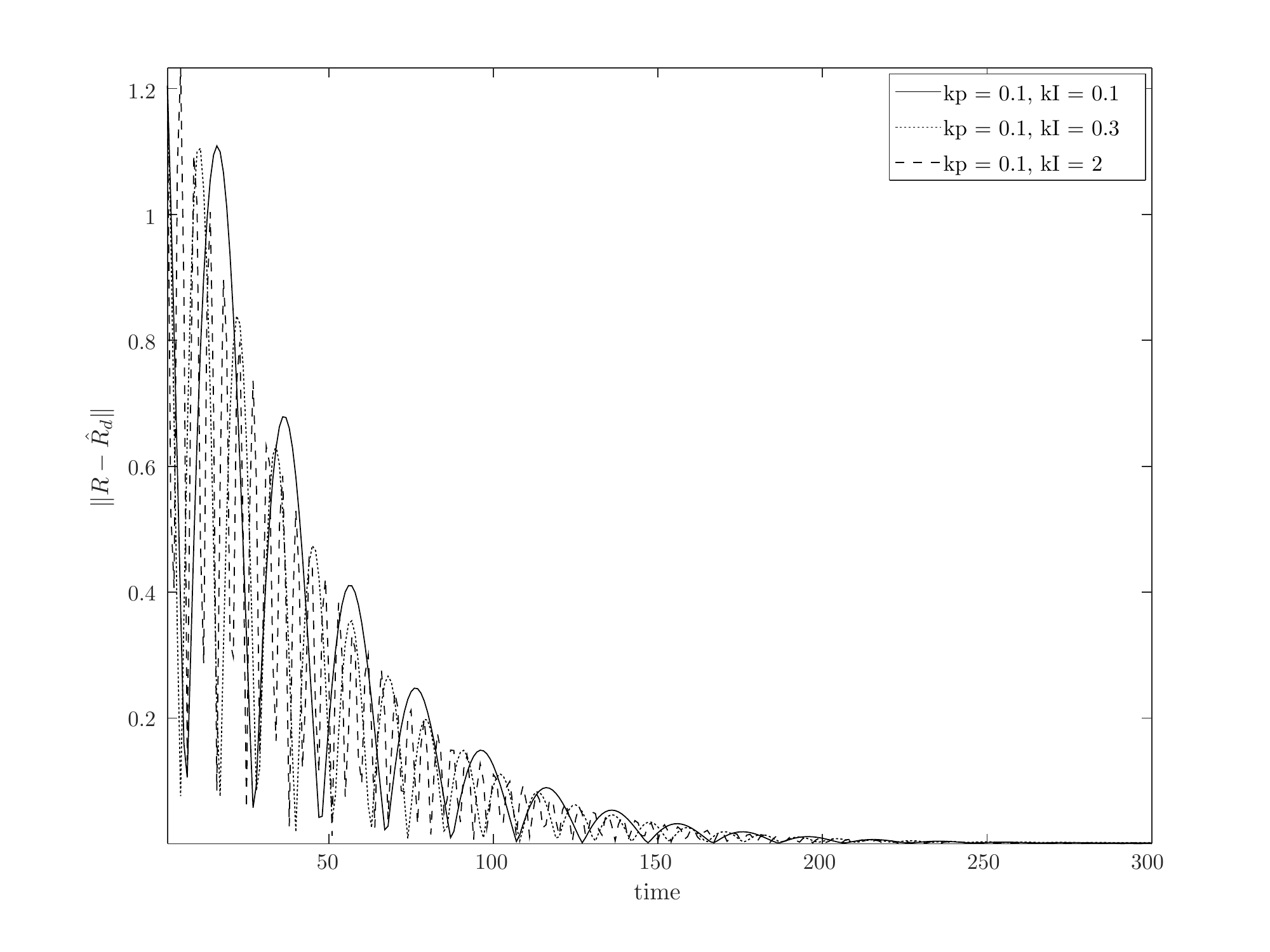}
	\caption{\label{fig:err_kI}Variation of rate of convergence with $k_I$ ($k_p$ and $k_e$ constant) ($\dt = 0.5s$)}
\end{figure}

}{
Figure \pref{fig:err} shows the variation of the error in the observed estimate with time.
An important feature to note is that the observer converges to the actual state eventually and the error in the observed state keeps reducing.
A few conclusions can be drawn from this:
\begin{itemize}
    \item Changing only the value of $k_e$ may not change the rate of convergence of the system.
    \item Increasing value of $k_p$ reduces the time taken by the system to reach the desired value and results in oscillations.
    However, increasing $k_p$ beyond a certain value also leads to the system never reaching the desired value and oscillating about a point with a finite error.
    \item Increasing value of $k_I$ increases the frequency of oscillations, but does not affect the time required to settle to the required value.
\end{itemize}
\begin{figure}[h]
    \centering
    \includegraphics[width=0.8\linewidth]{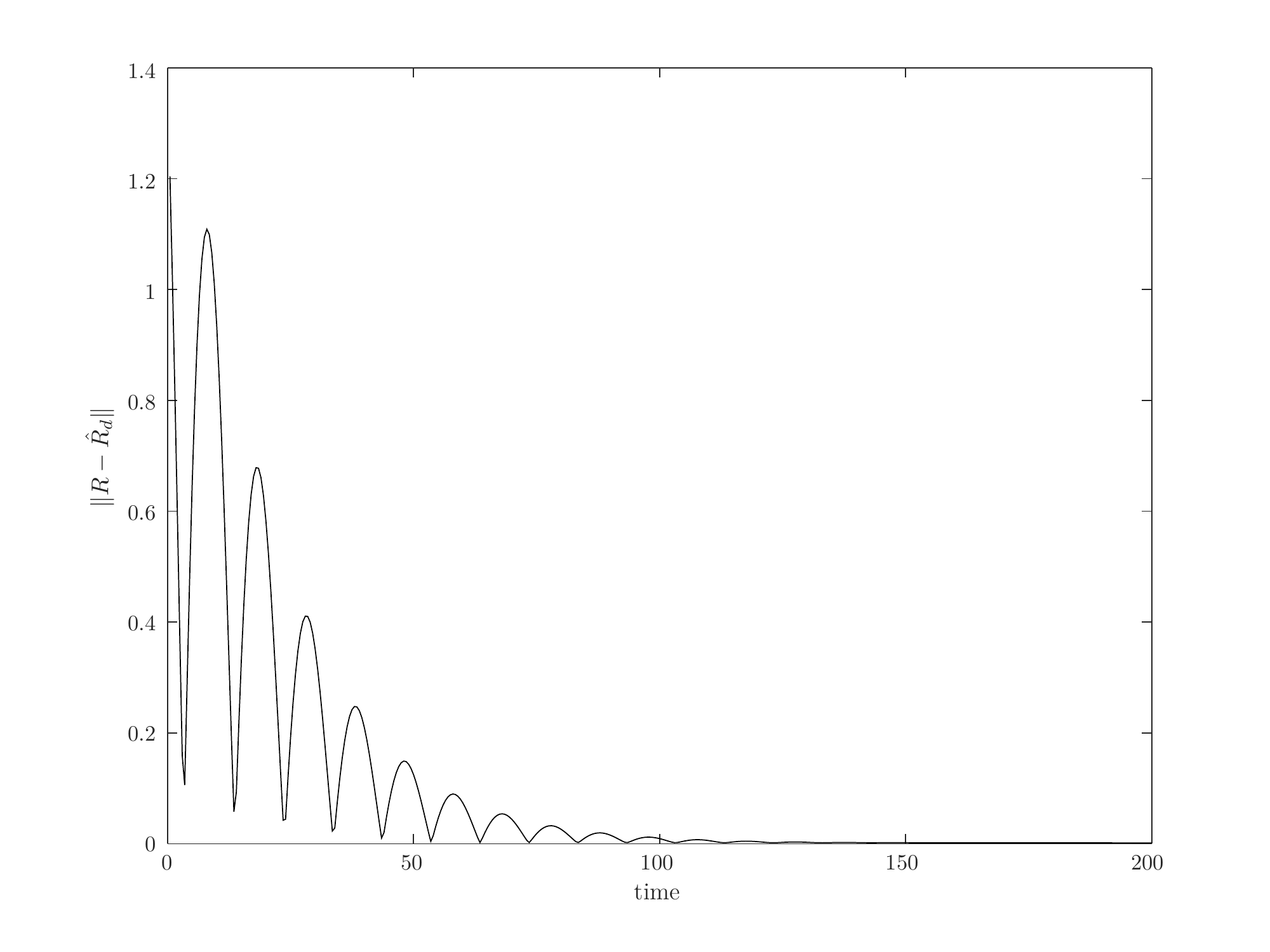}
    \caption{\label{fig:err}Evolution of frobenius norm of error between estimate and true value with time ($\dt = 0.5s$)}
\end{figure}
}

We change values of $k_p$ and $k_I$ to arrive at a faster converging observer for the system.
Refer to Figure \pref{fig:improved_con}.

\begin{figure}[h]
    \centering
    \includegraphics[width=0.8\linewidth]{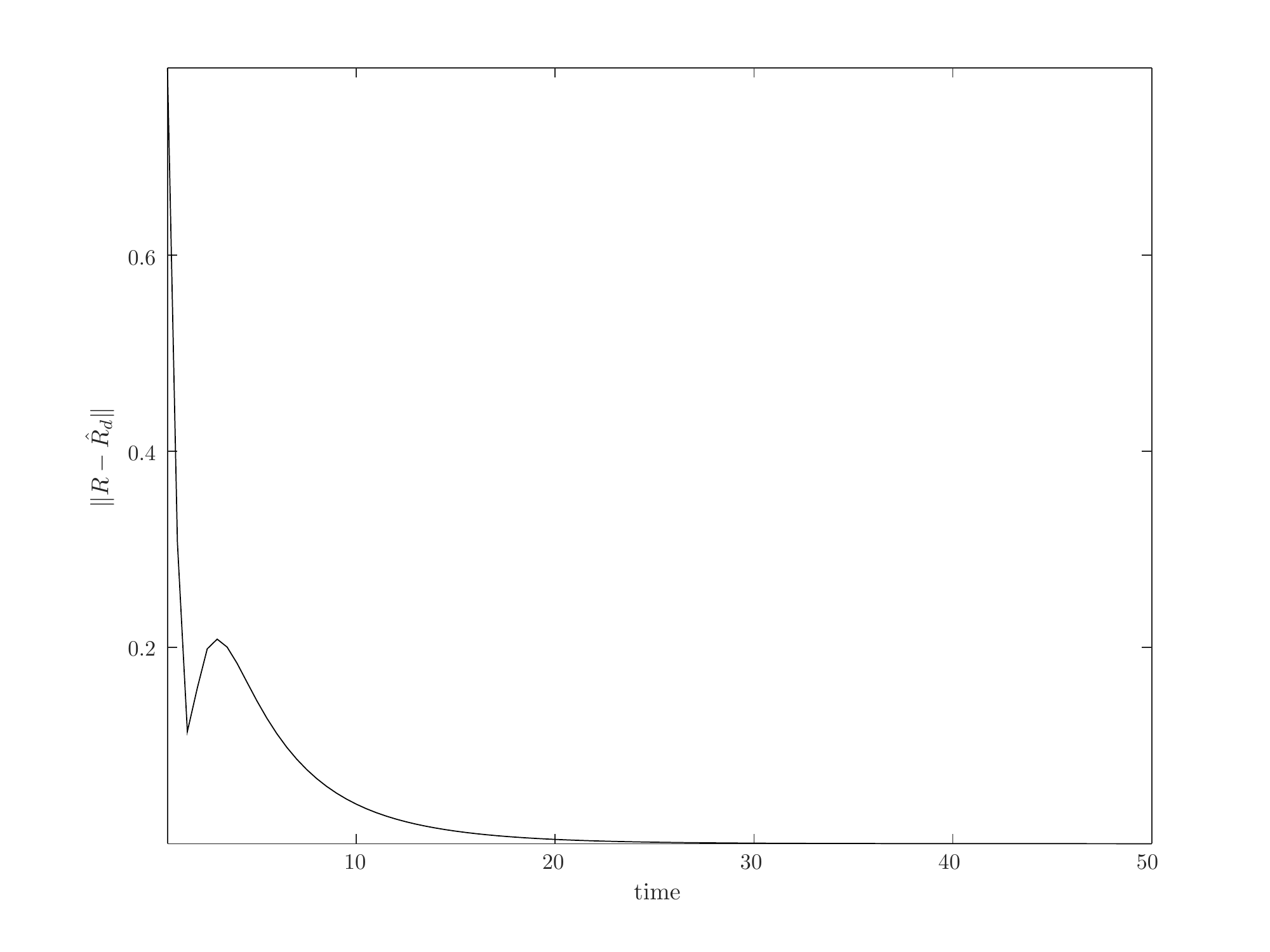}
    \caption{\label{fig:improved_con}Faster convergence with $k_p = 1$ and $k_I = 0.3$ ($\dt = 0.5s$)}
\end{figure}%}}}

\subsection{Convergence of system to the manifold}
% CONVERGENCE OF SYSTEM TO MANIFOLD{{{
We consider the effect of the term $k_e\hat{R}_d(k\mid k-1)(\hat{R}_d(k\mid k-1)^T\hat{R}_d(k\mid k-1) - I)\dt$ on keeping the system trajectory on the manifold by performing simulations ignoring the term and including the term.
We know that an element $A \in \SO3$ satisfies $\tr(A^TA - I) = 0 \Rightarrow \| A \|^2 - 3 = 0$, the norm used being the Frobenius Norm.
We use this property to check if the system converges to the estimate.
The final value attained by the system with feedback integrator term is 1.7321 as compared to the value 1.9218 attained by the system without feedback integrator term in 100s.
\begin{figure}[h]
	\centering
	\includegraphics[width=0.8\linewidth]{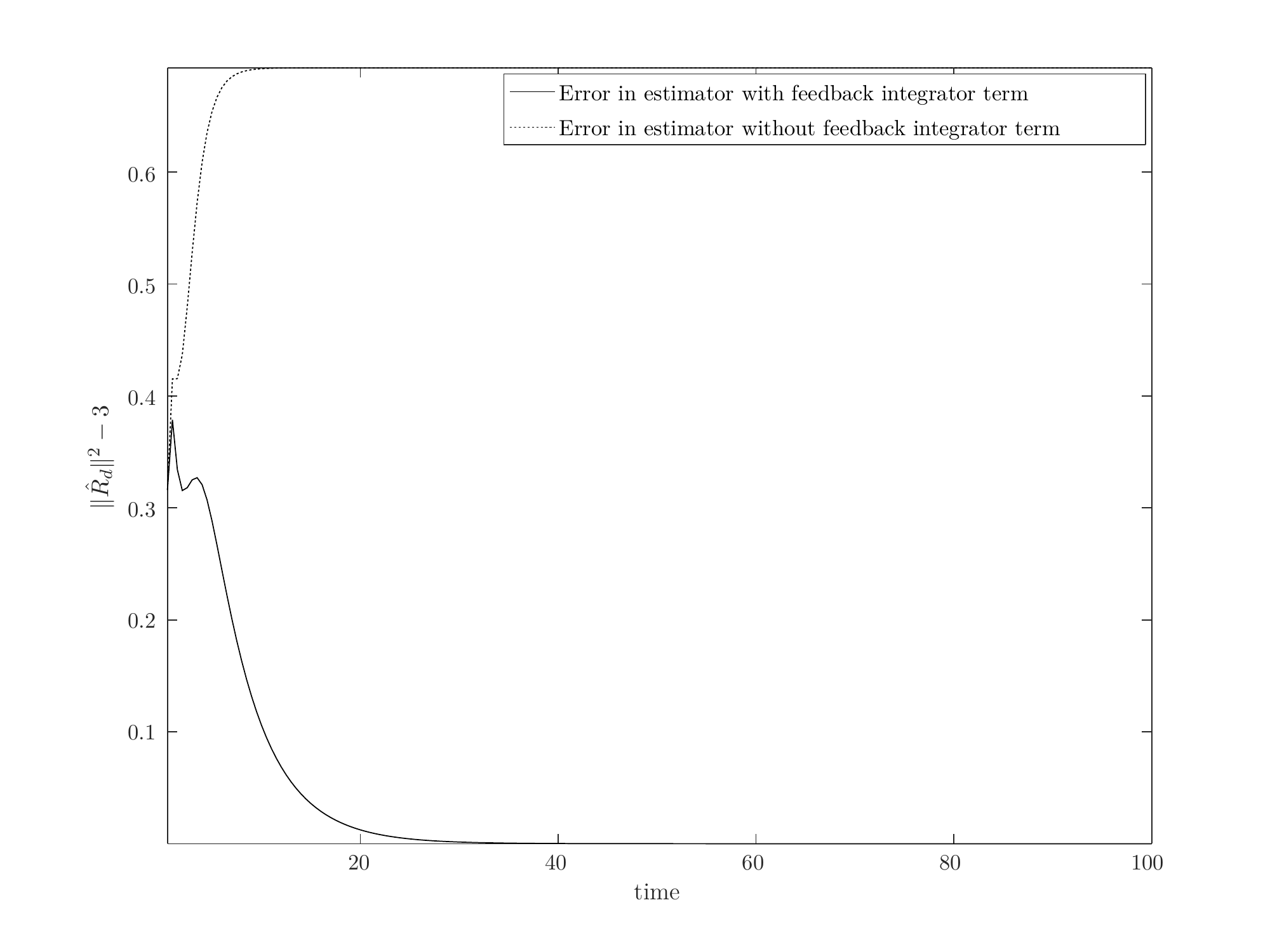}
	\caption{Convergence of the system to manifold\label{fig:conv_manifold} ($\dt = 0.5s$)}
\end{figure}

Moreover, without the $k_e$ term, the system does not converge to the manifold.
This can be seen from figure \pref{fig:conv_manifold_err}.
\begin{figure}[h]
	\centering
	\includegraphics[width=0.8\linewidth]{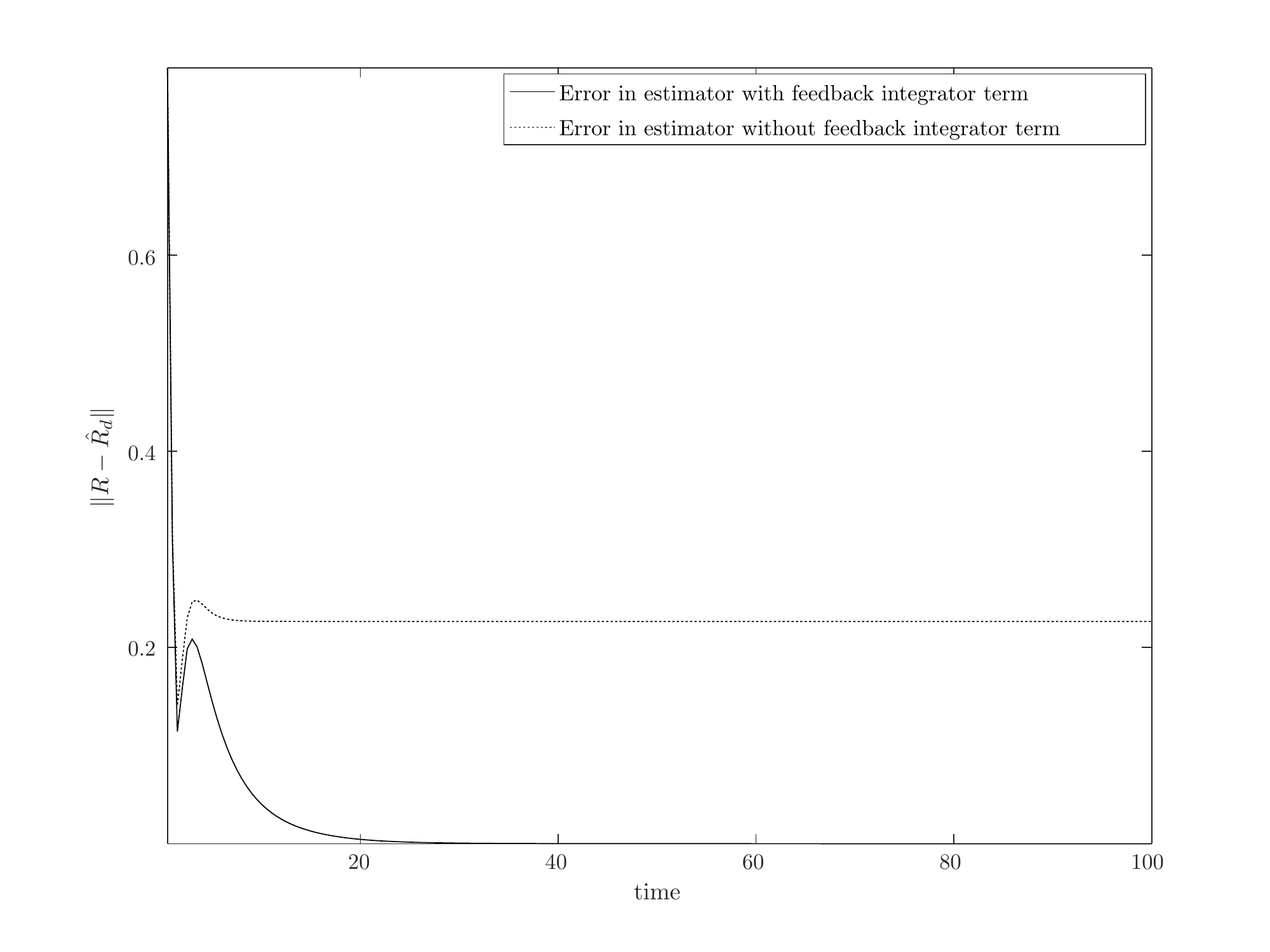}
	\caption{Effect of feedback integrator term on error between estimate and real value of system\label{fig:conv_manifold_err}}
\end{figure}

%}}}

\subsection{Response of system in presence of noise}
% NOISE CHARACTERISTICS{{{
We introduced noise of high frequency (as compared to the discretization time step) as a sinusoid into the measurements to study the effect on the performance of the observer.
The sinusoid has a frequency of 159 Hz and amplitude of 0.1. The measurements are now of the form
\begin{align*}
	R_k^y & = R(k\dt)\exp\left(0.1\sin(\omega k\dt)\begin{pmatrix}1\\ 1\\ 1\end{pmatrix}\right)\\
		\Omega_k^y & = \Omega(k\dt) + 0.1\sin(\omega k\dt)\begin{pmatrix}1\\ 1\\ 1\end{pmatrix}
\end{align*}

For the system performance in presence of noise, refer to Figure \pref{fig:noise_omega} and Figure \pref{fig:noise_R}. It is seen that the system does converge to a ball around the true value indicating that the system is robust to noise in measurements.
\begin{figure}[h]
    \centering
    \subfloat[Noise in $\Omega$ of magnitude 0.1\label{fig:noise_omega}]{\includegraphics[width=0.8\linewidth]{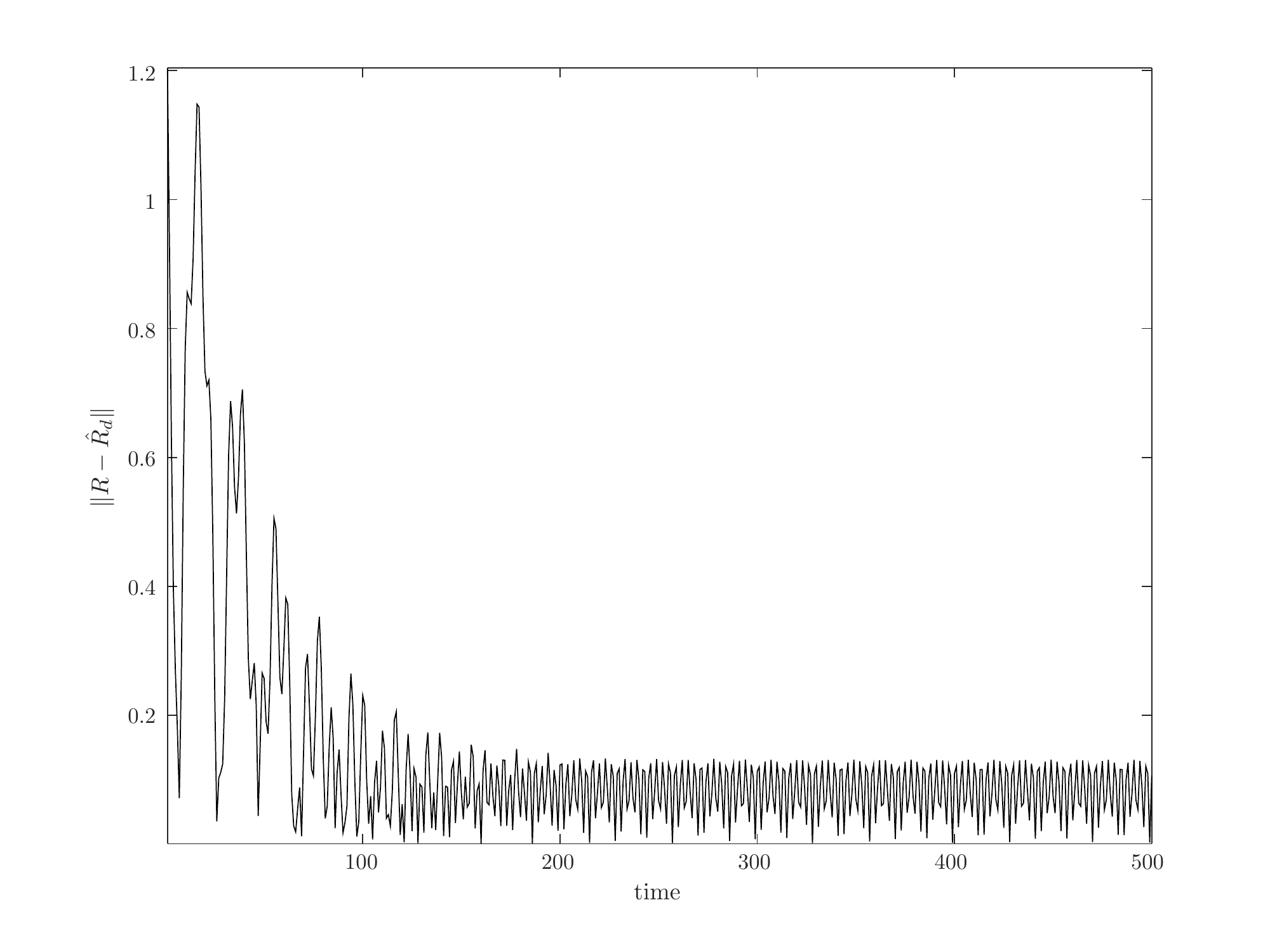}}\\
    \subfloat[Noise in R of magnitude 0.1\label{fig:noise_R}]{\includegraphics[width=0.8\linewidth]{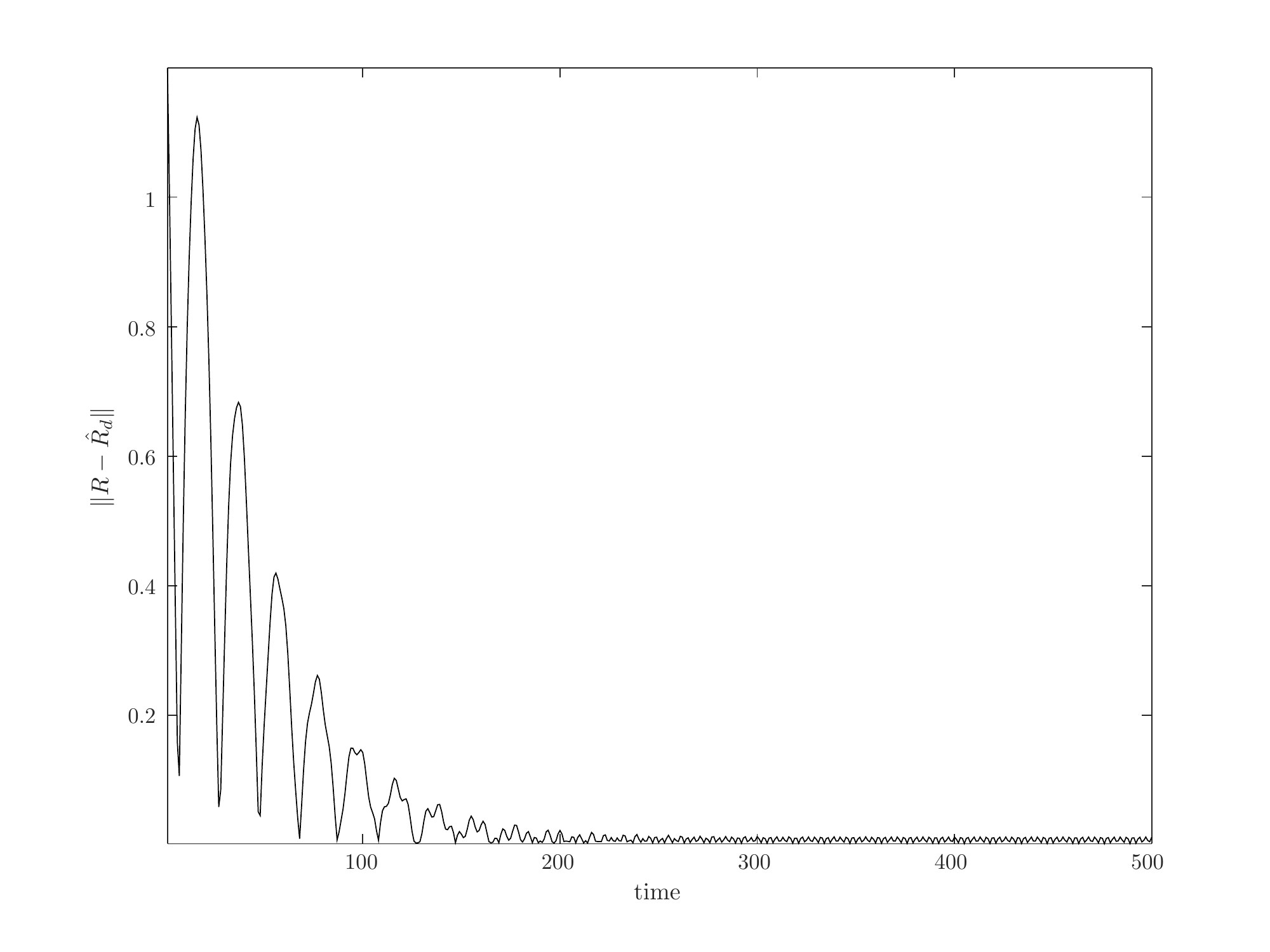}}
    \caption{Rate of convergence in presence of noise ($\dt = 0.5s$)}
\end{figure}%}}}

\iftoggle{arxiv}{
\subsection{Variation of estimate error with discretization interval $\dt$}
% remove and keep in arxiv version.
% ERROR CHARACTERISTICS{{{
We see that the amplitude of the error plots reduces with time and they have an oscillatory nature.
Hence, we would like to know the nature of reduction of this amplitude of oscillations.
For this purpose, we take the local maxima (or the crest of these oscillations) and consider a curve passing through these points.
The following plot shows the log of the local maximas of error with time.
The simulations are run for 1000s.
\begin{figure}[h]
    \centering
    \includegraphics[width=0.8\linewidth]{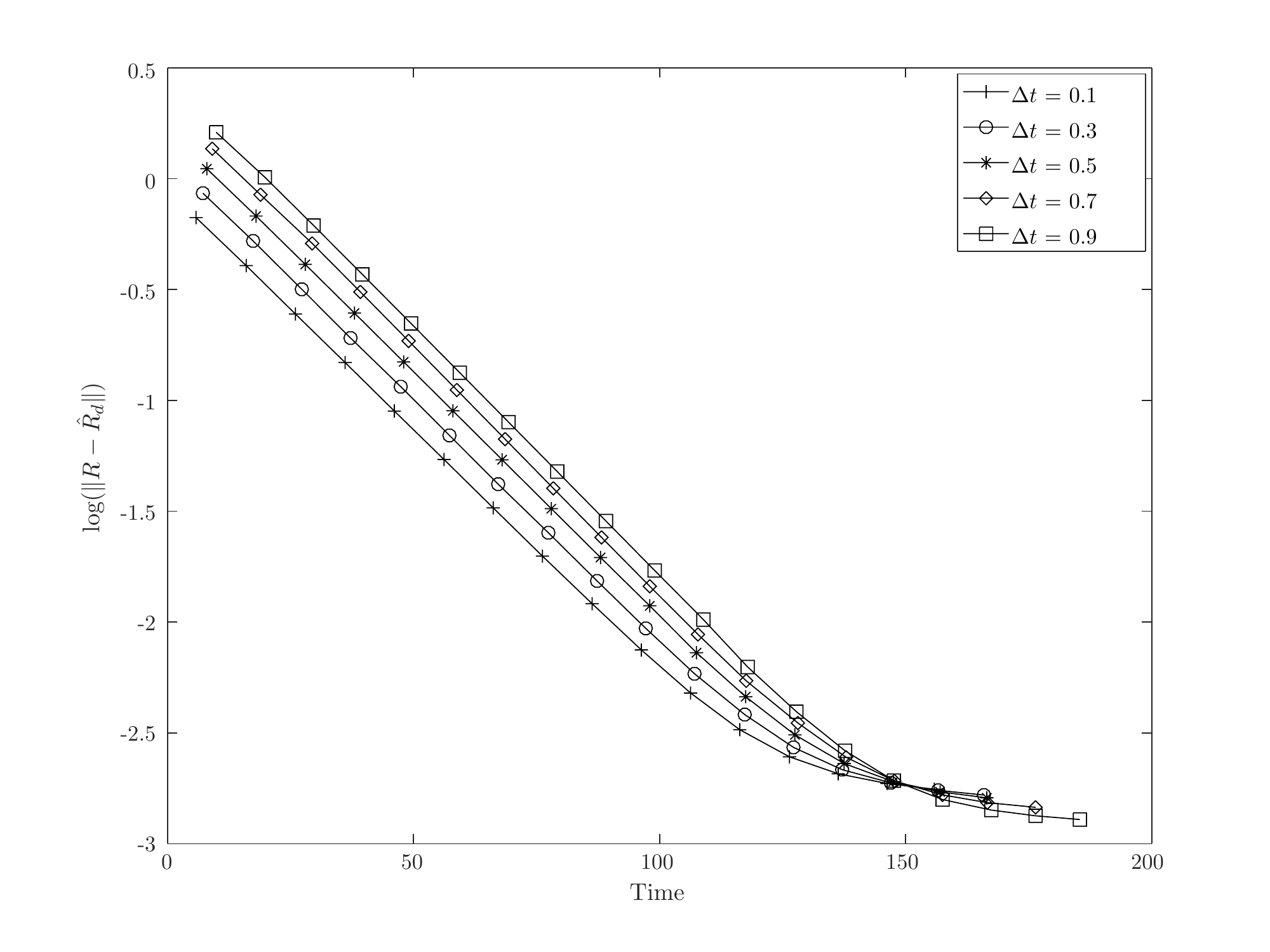}
    \caption{Variation of local maxima of error with time}
\end{figure}
The value of the error keeps on decreasing and is quite small around 200s in each of the cases, and the plots do not show any peaks.
However, due to the discrete nature of the system, the estimated value never exactly matches the actual value.%}}}

\subsection{Comparison between the proposed observer and Euler discretization}
% remove and keep in arxiv version.
% COMPARISON TO EULER{{{
Consider the Euler discretization of the Mahony observer \pref{eq:mahony_passive_filter} given by
\begin{subequations}
\begin{align}
	\omega_{eu, k} &= vex(\mathbb{P}_a(\hat{R}_{eu,k}^TR^y_k))\\
	\hat{R}_{eu,k+1} &= \hat{R}_{eu,k}\left(I + (\Omega^y_k - \hat{b}_{eu,k} + k_p\omega_{eu,k})_{\times}\dt\right)\\
	\hat{b}_{eu,k+1} &= \hat{b}_{eu,k} - k_I\omega_{eu,k}\dt
\end{align}
\end{subequations}

We compare the estimate arrived through this with the observer proposed in this article.
The comparative plot for the same is shown in figure \pref{fig_comp_eu_pro}.
\begin{figure}[h]
    \centering
    \includegraphics[width=0.8\linewidth]{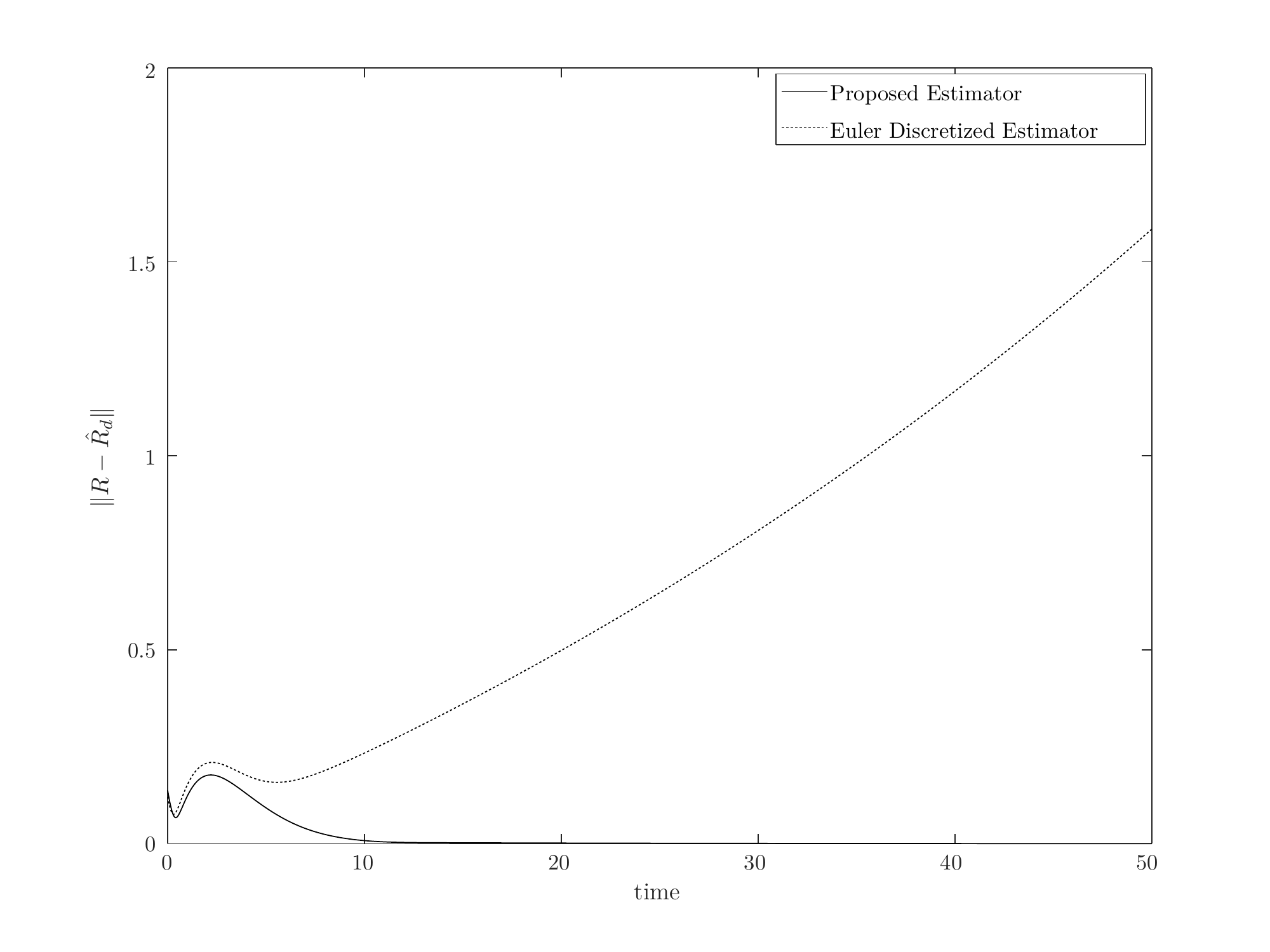}\\
    \caption{Comparison of convergence between Euler discretization and Proposed observer at $\dt = 0.01$s\label{fig_comp_eu_pro}}
\end{figure}

The system does not converge for $\dt = 0.5$s, which is the discretization time taken for all the other simulations presented in this report. Hence, the proposed observer is a huge improvement over standard Euler discretization.
}{}

	\section{Experimental Results}
		\label{sec:experiment}
		% SIMULATIONS ON ACTUAL DATA
We consider offline simulations performed on data acquired through experiment.
The experiment uses an ARdrone which is flown for approximately 100s.
It is used to capture data using an onboard magnetometer, a gyroscopic sensor and an accelerometer.
This data is sent to the computer.
Simultaneously, data using Vicon measurement systems is collected, which has a much higher accuracy and is considered the ``true'' value of the state.
However, due to human errors, the inertial frame of reference in which the Vicon measurements are collected and one in which the onboard sensors collect data are different.
To correct for this error, the system is kept at rest for some amount of time initially and measurement from both the systems are collected.
These measurements are used to calculate the rotation error in the stationary frames of references and correct the true measurements for comparison during simulations.
The time difference between successive data points collected using onboard sensor is approximately 0.02s.

This data set was first used with the Euler-discretization based observer
\begin{subequations}
\begin{align}
	\omega_{eu, k} &= vex(\mathbb{P}_a(\hat{R}_{eu,k}^TR^y_k))\\
	\hat{R}_{eu,k+1} &= \hat{R}_{eu,k}\left(I + (\Omega^y_k - \hat{b}_{eu,k} + k_p\omega_{eu,k})_{\times}\dt - k_e(\hat{R}_{eu,k}^T\hat{R}_{eu,k} - I)\dt\right)\\
	\hat{b}_{eu,k+1} &= \hat{b}_{eu,k} - k_I\omega_{eu,k}\dt
\end{align}
\end{subequations}

For the Euler-discretization based Mahony observer, we choose data points at time steps of 0.04s.
The resulting error plot between the true value of the rotation matrix and the estimated value using the observer is shown in Figure \pref{fig:real_eu}.
\begin{figure}[h]
	\centering
	\includegraphics[width=0.8\linewidth]{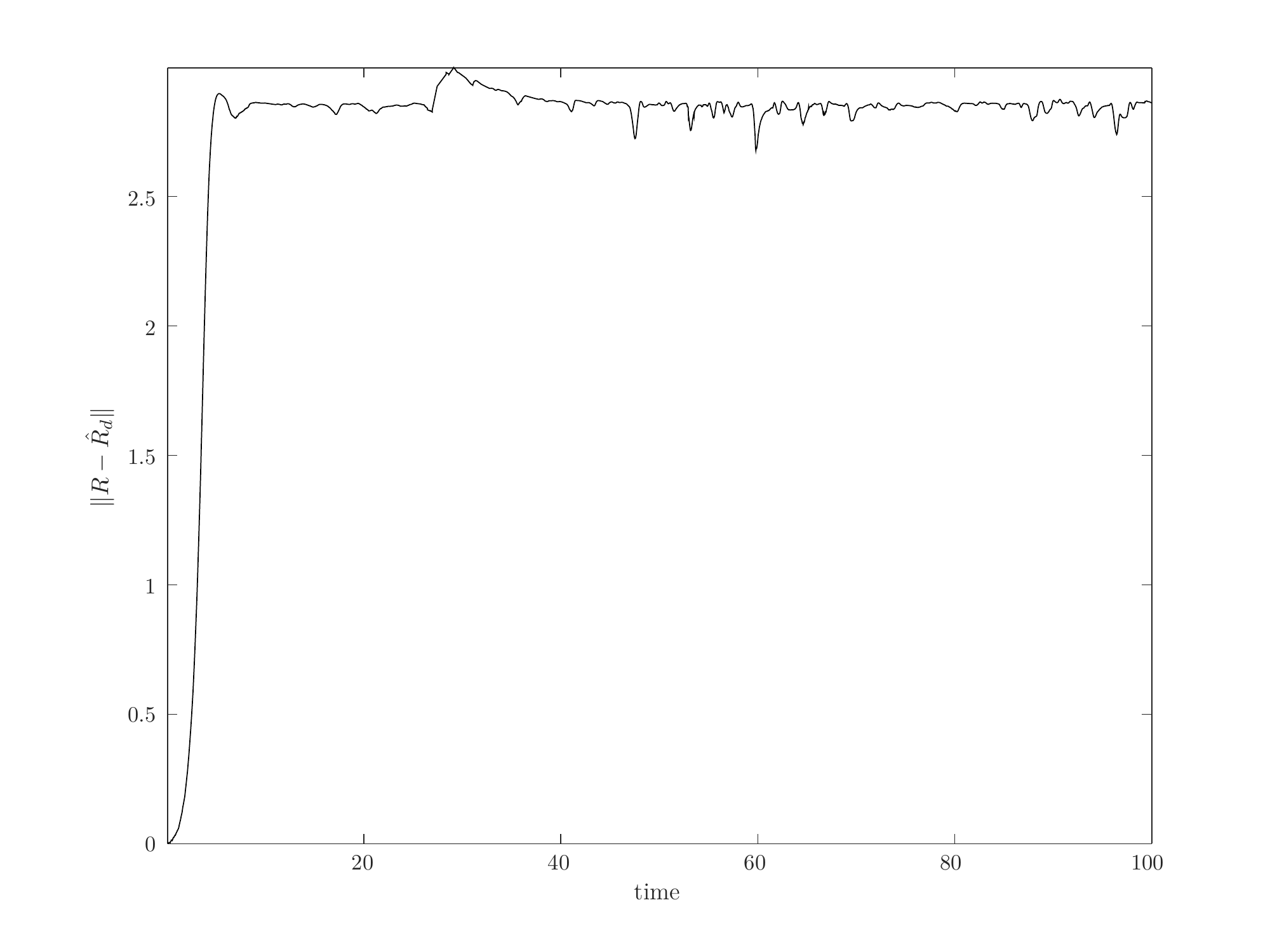}
	\caption{Estimate error with time using Euler discretization \label{fig:real_eu}}
\end{figure}
It can be seen from the plot that the Frobenius norm of the error ends up constant around a value of 3.
In an ideal scenario, we would like the estimate to mimic the true value, hence we would like the norm of the error to be 0.
Hence, this estimator is not a ``good'' estimator of the system.

Then the data set was used with the proposed observer \pref{eq:discrete_observer} with $\dt = 0.2s$.
The data for this is acquired from the collected data by taking the time instants which are closest to multiples of 0.2s, i.e. $0.2s, 0.4s, \ldots$.
The error between the estimate and the true value is shown in Figure \pref{fig:real_proposed}.
\begin{figure}[h]
	\centering
	\includegraphics[width=0.8\linewidth]{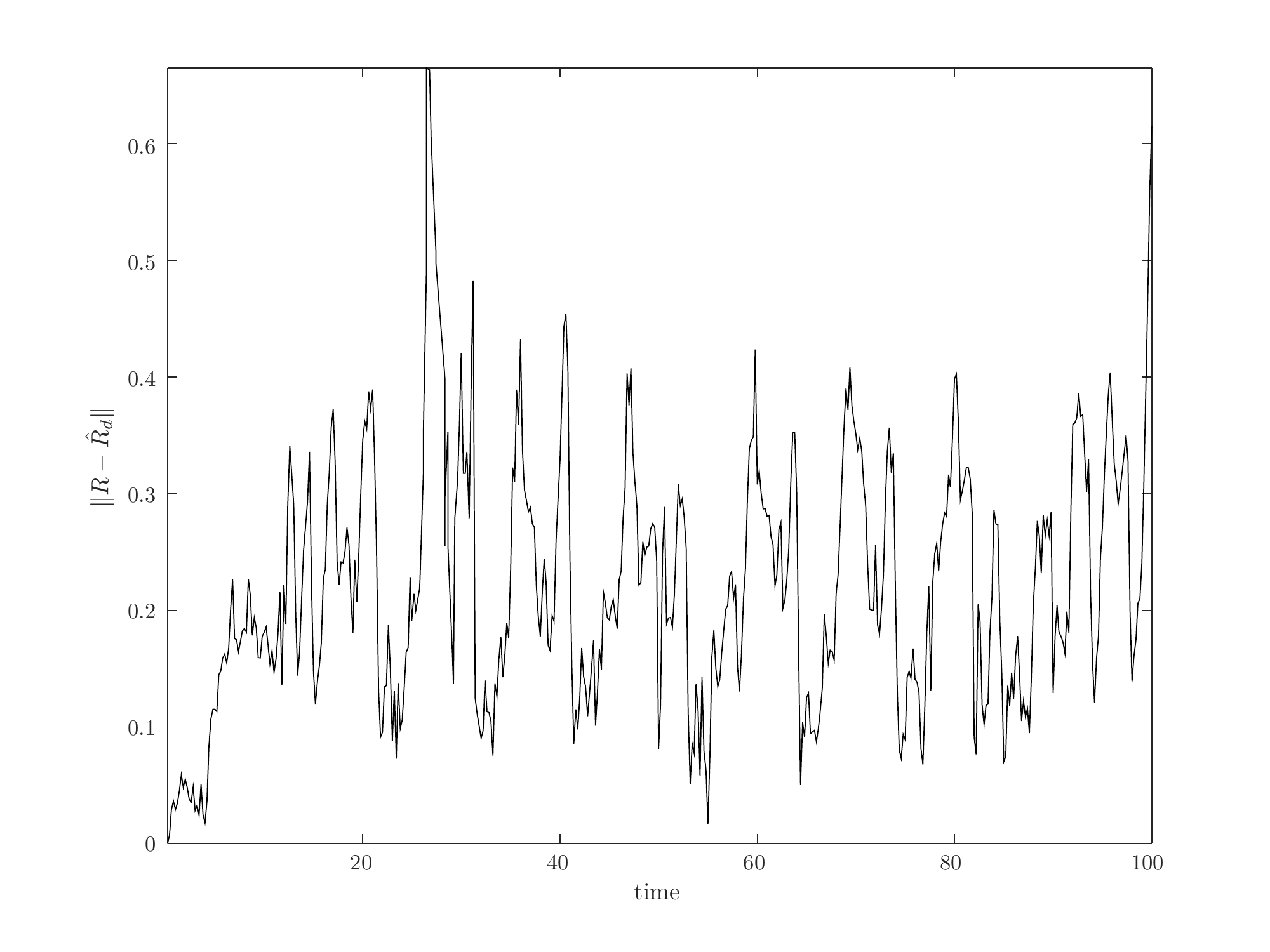}
	\caption{Estimate error with time using the proposed observer \label{fig:real_proposed}}
\end{figure}
We can see that the maximum error of the mean of the data here is of magnitude 0.6 and it reduces after reaching this value.
Moreover, the mean of this data would be around 0.3.
Since the error in the observed values is less compared to that arrived at by Euler-discretisation based observer, and with much sparser measurements, we conclude that this observer gives a better estimate of the system than Euler discretisation.

	\section{Convergence vs Complexity}
		\label{sec:complexity}
		% COMPLEXITY VS CONVERGENCE
Given that our proposed observer uses an exponential prediction term, one would believe that it would require significantly higher level of computing time as compared to a simple Euler observer.
However, the benefit of our observer is pronounced when the system in question has sparse measurements (larger $\Delta T$).
We now show two simulations comparing the Euler discretized observer with feedback integrator and the proposed observer on the time required by the observers to converge to the real value as measurements get more infrequent.
\begin{figure}[h]
	\centering
	\subfloat[Euler discretization, $\dt = 0.001s$]{\includegraphics[width=0.45\linewidth]{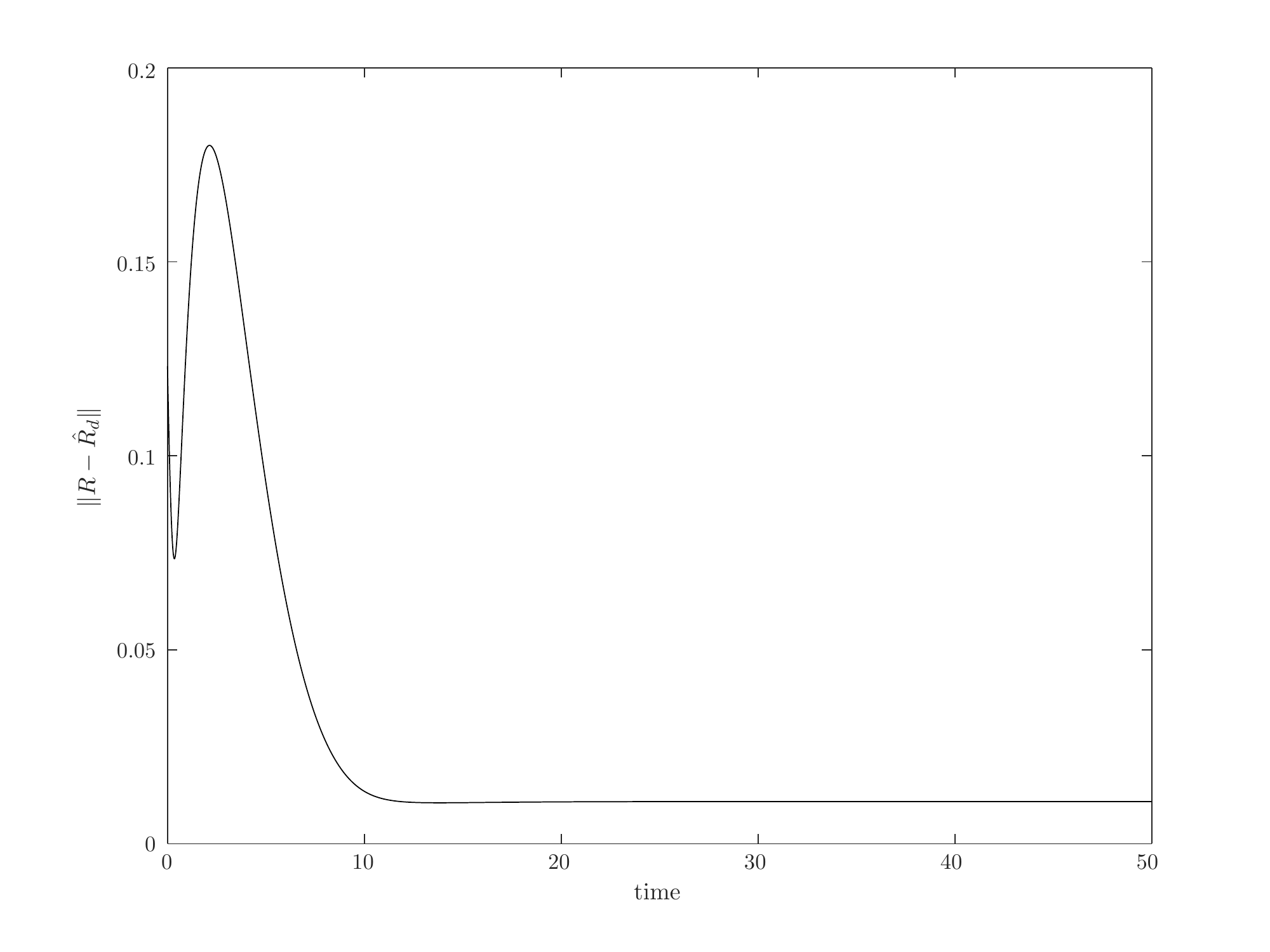}}\quad
	\subfloat[Proposed observer, $\dt = 0.5s$]{\includegraphics[width=0.45\linewidth]{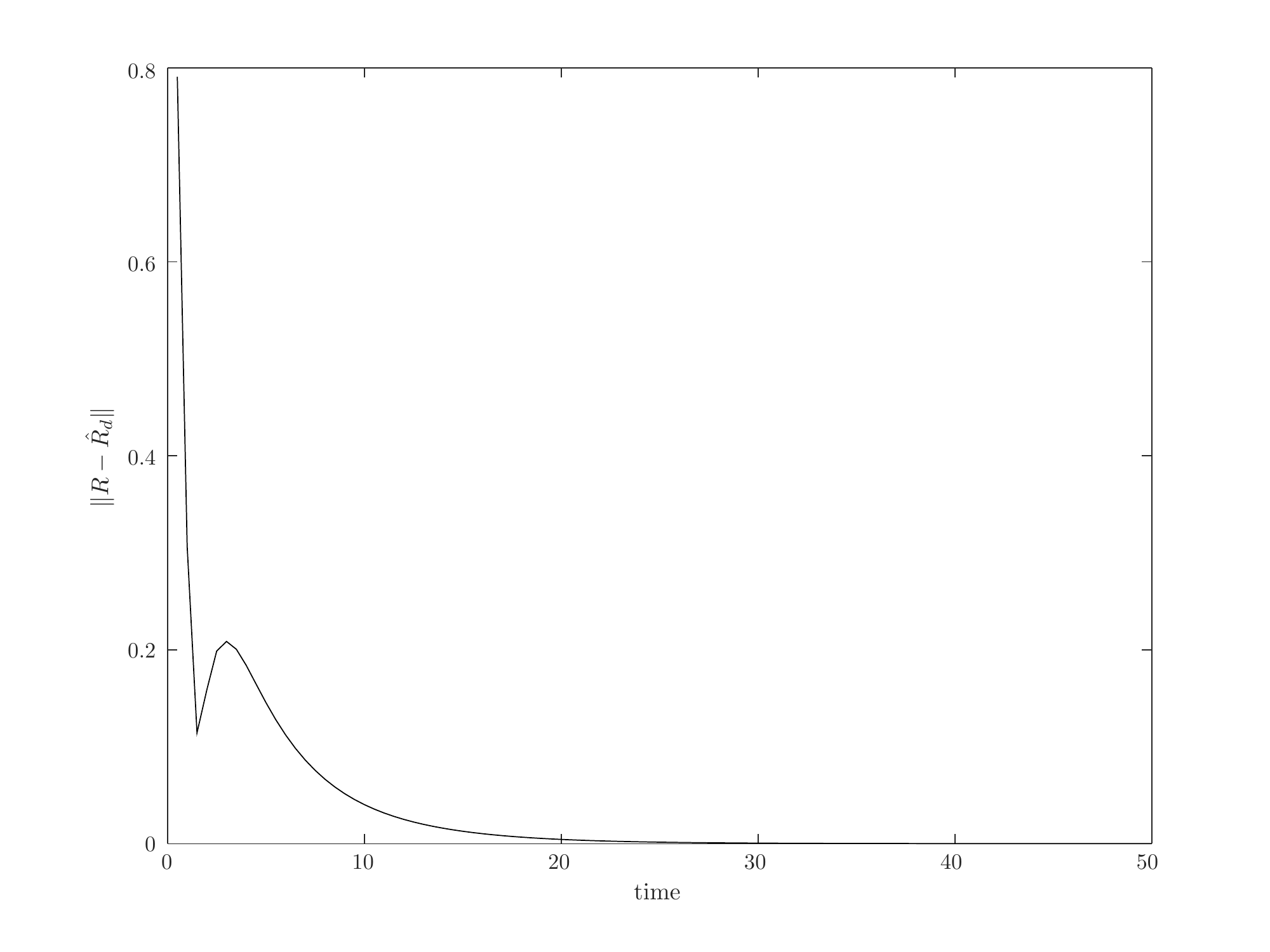}}
	\caption{Frequency of data availability vs convergence comparison between Euler discretization and proposed observer\label{fig:conv_vs_compl}}
\end{figure}

As can be seen in Figure \pref{fig:conv_vs_compl}, even with a comparatively higher time step, the proposed observer performs significantly better than Euler discretization.
Moreover, the Euler discretization takes over 30-40s in simulations, while the proposed observer takes about 0.3s.
Hence, we can definitely conclude that the proposed observer is far superior as compared to Euler discretization, and whenever the measurements are sparse, the proposed observer is a much better choice as compared to Euler discretization even including feedback integrator.

As a concluding remark, the implementor is free to discretize the exponential using Taylor series expansion if he wishes to reduce the complexity of the observer, while not compromising on convergence results.
As an example of this, a simulation with Taylor expansion till the second term is shown in figure \pref{fig:taylor}. The time taken for this is 0.8s.
\begin{figure}[h]
	\centering
	\includegraphics[width=0.8\linewidth]{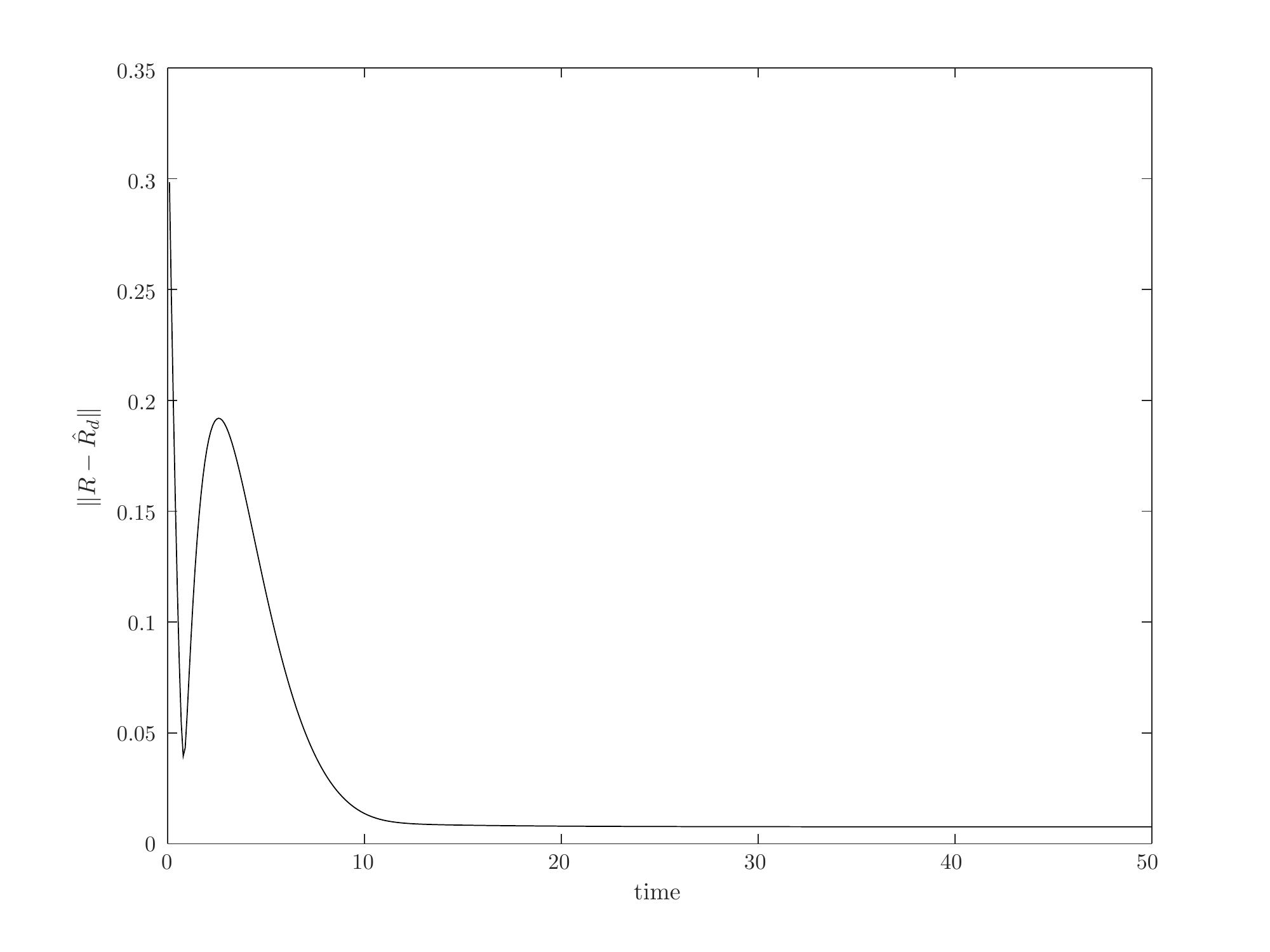}
	\caption{Taylor discretized predictor step\label{fig:taylor}}
\end{figure}

	\iftoggle{arxiv}{}{
	\section{Alternate Applications}
		\label{sec:duality}
		% PATH PLANNING == CONTROLLER DESIGN == OBSERVER DESIGN{{{
Although the proposed algorithm targets an estimation problem, the same approach can be used to design a controller or to solve a path planning problem.

\subsection{Controller design: Stabilization}
The notion of controller design for stabilization is that the system is to be stabilised about a certain equilibrium point. Assume that the system to be stabilised is given by 
\begin{align}
	\dot{R} = R\Omega_\times\label{eq_controller_system}
\end{align}
where $R$ is the state and $\Omega$ here is the control input. 
Further, we have measurements at discrete instants of time as
\begin{align*}
	R^y_k = R(k\dt)
\end{align*}
Assume that we wish to stabilise this system about $p \in \SO3$ in a discrete setting. So we compute the control satisfying the discrete equation (source: \cite{phogat})
\begin{align*}
	R_{k+1} = R_{k}\Omega_c
\end{align*}
Our philosophy in designing such a control is to make the system \pref{eq_controller_system}
track the following system given by the state and input two-tuple
 $(R_1, \Omega_1)$, and  which has the state update laws:
\begin{align*}
	R_1 & = p\\
	\Rightarrow \dot{R}_1 & = 0 = R_1
	\begin{pmatrix}0\\0\\0
	\end{pmatrix}_\times\label{eq_controller_dummy_system}\numberthis
\end{align*}

Since the latter system is the constant $(p, 0)$, it follows from the observer results of this report that $R$ will eventually converge to the point $p$. We now show the modification required to do so:

We can modify the designed observer in the following way to stabilize it about the given point.
From the dynamics, it follows that $\Omega_1$ is $0$ at all times.
Similarly, the initial state of the system \pref{eq_controller_dummy_system} can be taken to be $p$ and the measurement of the system at the $k^{th}$ epoch can be taken to be $p$.
This leads to a controller which is stabilised around the specified point.
This translates to the observer system being modified as

The predictor step:  $t \in [(k-1)\dt, k \dt[$
\begin{align*}
	\hat{R}_d(k\mid k-1) = \hat{R}_d(k-1 \mid k-1)\exp(-\hat{b}_d(k-1)_\times\dt), \quad \hat{R}_d(0 \mid 0) = R_0
\end{align*}
now translates to
\begin{align*}
	x_d &= \exp(-\hat{b}_d(k-1)_\times\dt)
\end{align*}
The corresponding corrector step terms at $t = k\dt$ are
\begin{align*}
	\omega_k &= vex(\mathbb{P}_a(\hat{R}_d(k\mid k-1)^Tp))
\end{align*}
which translates to 
\begin{align*}
    \omega_k = vex(\mathbb{P}_a(x_d^T(R^y_k)^Tp))
\end{align*}
The term
\begin{align*}
	\begin{split}
        &\hat{R}_d(k\mid k) = \hat{R}_d(k\mid k-1) + \hat{R}_d(k\mid k-1)k_p\omega_{k_\times}\dt \\&\quad- k_e\hat{R}_d(k\mid k-1)(\hat{R}_d(k\mid k-1)^T\hat{R}_d(k\mid k-1) - I)\dt  
	\end{split}
\end{align*}
translates to 
\begin{align*}
    \Omega_c(k) &= x_d[I + k_p\omega_{k_\times}\dt - k_e(x_d^T(R^y_k)^TR^y_kx_d - I)\dt]
\end{align*}
The term
\begin{align*}
	\hat{b}_d(k) &= \hat{b}_d(k-1) + k_b \omega_k\dt, \quad \quad \quad \quad \quad b_d(0) = 0
\end{align*}
remains the same.

We use this to calculate $\Omega_c$, the controlled input to the system.
However, to account for the measurement to be taken at the $k$th instant, we rewrite the equations in the form given below.
The algorithm to calculate $\Omega_c$ will be
\begin{align*}
	\hat{b}_d(k) &= \hat{b}_d(k-1) + k_b \omega_k\dt\\
    x_d &= \exp(-\hat{b}_d(k)_\times\dt)\\
	\omega_k &= vex(\mathbb{P}_a(x_d^T(R^y_k)^Tp))\\
	\Omega_c(k) &= x_d[I + k_p\omega_{k_\times}\dt - k_e(x_d^T(R^y_k)^TR^y_kx_d - I)\dt]
\end{align*}
Due to the designed observer converging to the system dynamics \pref{eq:sys}, this controller ensures that the state converges to the required equilibrium point exponentially.

\subsection{Path Planning}
Similar to controller design, the same observer equations can be modified to follow a path.
We assume that the system $\dot{R} = R \Omega_{\times}$ should follow a certain 
desired path $f(t) \in \SO3$ and we wish to obtain the $\Omega(\cdot)$ trajectory to
do so. We proceed as follows. Differentiating we obtain
\begin{align*}
	\dot{R}(t) = \dot{f}(t) = f(t)f^T(t)\dot{f}(t) = R(t)f^T(t)\dot{f}(t) = R(t)g(t)
\end{align*}
where $g = f^T\dot{f}$.
Since this equation is in the form we have for our system, we choose $\Omega_k^y$ based on the value of $g$ at the $k^{th}$ epoch. This enables us to use our designed observer as a path planning controller.
The path planning controller equations then are:

Predictor step:  $t \in [(k-1)\dt, k \dt[$
\begin{align*}
	\hat{R}_d(k\mid k-1) = \hat{R}_d(k-1\mid k-1)\exp(\hat{\Omega}_d(k-1)_\times\dt)
\end{align*}

Corrector step: at $t = k\dt$
\begin{align*}
	\omega_k &= vex(\mathbb{P}_a(\hat{R}_d(k\mid k-1)^TR^y_k))\\
	\begin{split}
		\hat{R}_d(k\mid k) &= \hat{R}_d(k\mid k-1) + \hat{R}_d(k\mid k-1)k_p\omega_{k_\times}\dt \\&\quad- k_e\hat{R}_d(k\mid k-1)(\hat{R}_d(k\mid k-1)^T\hat{R}_d(k\mid k-1) - I)\dt
	\end{split}\\
	\hat{b}_d(k) &= \hat{b}_d(k-1) + k_b \omega_k\dt\\
	\hat{\Omega}_d(k) &= g(k\dt) - \hat{b}_d(k)
\end{align*}
Hence, we have that for small enough $\dt$, for all $k \in \mathbb{N} > M$,
\begin{align*}
    \| \hat{R}_d(k \mid k) - f(k\dt) \| < \epsilon
\end{align*}

		}

	\section{Acknowledgements}
	    Soham Shanbhag would like to acknowledge Prof. Debasish Chatterjee and Prof. Srikant Sukumar, and his colleague, Mishal Assif P K, for their inputs in this work.
The authors would also like to thank D. H. S. Maithripala for his suggestions regarding the estimator design, and Aseem V. Borkar, for his help  in collecting the data for the experiment.

	\appendix

	\section{Feedback Integrators}\label{app:FI}
\begin{theorem}\label{theo:FI}\cite{Chang_FI}
	Consider a dynamical system on an open subset U of $\mathbb{R}^n$:
    \begin{align}
        \dot{x} = X(x),\label{eq:FI_system}
    \end{align}
    where X is a $C^1$ vector field on U.
    Let us make the following assumptions:
    \begin{enumerate}
        \item[A1.] There is a $C^2$ function $V: U \to \mathbb{R}$ such that $V(x) \geq 0$ for all $x \in U, V^{-1}(0) \neq \phi$, and
        \begin{align}
            \nabla V(x)\cdot X(x) = 0
        \end{align}
        for all $x \in U$.
        \item[A2.] There is a positive number c such that $V^{-1}([0, c])$ is a compact subset of U.
        \item[A3.] The set of all critical points of V in $V^{-1}([0,c])$ is equal to $V^{-1}(0)$.
    \end{enumerate}
    Then, for the system
    \begin{align}
        \dot{x} = X(x) - \nabla V(x)\label{eq:modified_FI_system}
    \end{align}
    every trajectory starting from a point in $V^{-1}([0,c])$ stays in $V^{-1}([0,c])$ for all $t \geq 0$ and asymptotically converges to the set $V^{-1}(0)$ as $t \to \infty$.
    Furthermore, $V^{-1}(0)$ is an invariant set of both \pref{eq:FI_system} and \pref{eq:modified_FI_system}.\\
    It should be noted that both the vector fields coincide on $V^{-1}(0)$.
\end{theorem}

	\section{Convergence of the discretization of the modified dynamical system}\label{app:disc_conv}

The authors in \cite{Chang_FI} claim that the discrete time dynamical system derived as a one step integrator from a continuous time dynamical system extended by the Feedback Integrator to the ambient Euclidean space preserves the Manifold and the first integrals of the system. However, the proof of this claim is not shown. For the requirements of this project, this is important. Hence, we present a proof of the same here.

Consider the continuous time dynamical system
\begin{align}
    \dot{x} = X(x), \quad x \in U
\end{align}
where $U$ is an open subset of $\R^n$.\\

The extension of this system to the ambient Euclidean space is given by
\begin{align}
    \dot{x} = X(x) - \nabla V(x), \quad x \in \R
\end{align}
where $V: U \to \R$ is a $C^2$ function following the following assumptions:
\begin{itemize}
\item $V(x) \geq 0 ~ \forall ~ x \in U$, $V^{-1}(0) \neq \phi$
\item $\nabla V(x)\cdot X(x) = 0 ~ \forall ~ x \in U$
\item $\exists ~ c > 0$ such that $V^{-1}([0,c])$ is a compact subset of $U$
\item The set of all critical points of V in $V^{-1}([0,c])$ is equal to $V^{-1}(0)$
\end{itemize}
We consider the discretized system under any scheme as follows:
\begin{align}
    x_{k+1} = x_{k} + hf(x_k) - h\nabla V(x_k)\label{eq:disc_FI}
\end{align}
\begin{claim}: For a sufficiently small $h$ , the system \pref{eq:disc_FI} preserves first integrals and the manifold.
\end{claim}
\begin{proof}%{{{
    Consider the Lyanpunov like function $V(x_k)$ where $x_k \in V^{-1}([0,c])$..
    \begin{align}
        V(x_{k+1}) &= V(x_k + hf(x_k) - h\nabla V(x_k))\\
        &= V(x_k + h(f(x_k) - \nabla V(x_k)))
    \end{align}
    Using Taylor series expansion of $V$ at $x_k$, we get
    \begin{align}
        V(x_{k+1}) &= V(x_k) + \frac{\nabla V(x_k)}{1!}\cdot h(f(x_k) - \nabla V(x_k)) + o(h^2)\\
        &= V(x_k) - h|\nabla V(x_k)|^2 + o(h^2)
    \end{align}
    Hence,
    \begin{align}
        V(x_{k+1}) - V(x_k) = -h|\nabla V(x_k)|^2 + o(h^2) \leq 0 
        \quad \hbox{for sufficiently small $h > 0$} \label{eq:disc_FI_error}
    \end{align}
    Hence, $V^{-1}([0,c])$ is a positively invariant set of \pref{eq:disc_FI}. Since $V^{-1}([0,c])$ is also compact by assumption, hence discrete time LaSalle's Invariance Principle \cite{BulloDiscreteLasalla} can be applied to the system. Hence, the system converges to the largest invariant subset of $E = \{x \in V^{-1}([0,c]) | V(x_{k+1}) - V(x_k) = 0\}$.
    As is visible from equation \pref{eq:disc_FI_error}, this is dependent on $h$. Hence, we can say that $E = V^{-1}([0,\epsilon(h)])$.
    Hence, the discretized system converges to an epsilon-neighbourhood around the dynamics where the first integrals and the manifold are conserved, and the size of this neighbourhood can be defined by $h$.
\end{proof}%}}}

	\bibliographystyle{amsalpha}
	\bibliography{ref}

	\bigskip

\end{document}